\documentclass[journal]{IEEEtran}
\ifCLASSINFOpdf
\else
\fi
\usepackage{epsfig} 
\usepackage{mathrsfs}
\usepackage{verbatim}
\usepackage{graphicx}
\usepackage{float}
\usepackage{graphics} 
\usepackage{epsfig} 
\usepackage{times} 
\usepackage{amsmath} 
\usepackage{amssymb}  
\usepackage{bm}
 \usepackage{enumerate} 
\usepackage{cite}
\usepackage{graphicx}
\usepackage[numbers,sort&compress]{natbib}
\usepackage{booktabs}
\usepackage{color}
\newtheorem{Theorem}{Theorem}
\newtheorem{Remark}{Remark}
\newtheorem{Assumption}{Assumption}
\newtheorem{Lemma}{Lemma}
\newtheorem{Definition}{Definition}
\newtheorem{Corollary}{Corollary}

\newenvironment{proof}{{\emph {Proof.}}}{\hfill $\square$\par}

\begin{document}
%
\title{Invariance and Contraction in Geometrically Periodic Systems with Differential Inclusions}
%
%
%

\author{Chen~Qian and
        Yongchun~Fang,~\IEEEmembership{Senior Member,~IEEE}
\thanks{This work was supported by National Natural Science Foundation
of China under Grant 61873132, Grant 61633012. (Corresponding author: Yongchun
Fang) 

The authors are with the Institute of Robotics and Automatic Information
Systems, College of Artificial Intelligence, and also with the
Tianjin Key Laboratory of Intelligent Robotics, Nankai University, Tianjin
300353, China (e-mail: chainplain@mail.nankai.edu.cn; fangyc@nankai.edu.cn).
}
        }
%
%

\markboth{}%
{Shell \MakeLowercase{\textit{et al.}}: Bare Demo of IEEEtran.cls for IEEE Journals}
%



\def\Revised#1{{\textcolor[rgb]{0.2,0.2,0.65}{#1}}}
\def\Origin#1{{\textcolor[rgb]{0.65,0.2,0.2}{#1}}}

\maketitle

\begin{abstract}
The objective of this paper is to derive the essential invariance and contraction properties for the geometric periodic systems, which can be formulated as a category of differential inclusions, and primarily rendered in the phase coordinate, or the cycle coordinate.
First, we introduce the geometric averaging method for this category of systems, and also analyze the accuracy of its averaging approximation. 
Specifically, we delve into the details of the geometrically periodic system through the tunnel of considering the convergence between the system and its geometrically averaging approximation.
Under different corresponding conditions, the approximation on infinite time intervals can achieve certain accuracies, such that one can use the stability result of either the original system or the averaging system to deduce the stability of the other. 
After that, we employ the graphical stability to investigate the ``pattern stability'' with respect to the phase-based system.
Finally, by virtue of the contraction analysis on the Finsler manifold, the idea of accentuating the periodic pattern incremental stability and convergence is nurtured in the phase-based differential inclusion system, and comes to its preliminary  fruition in application to biomimetic mechanic robot control problem. 
\end{abstract}


%
\IEEEpeerreviewmaketitle

{

\section{Introduction}
Periodic systems and periodic motions exist ubiquitously in every field of science and everywhere in real life \cite{Farkas}, which can be observed in celestial movements, chemical reactions, ecological systems, robotic locomotions and so on and so forth.
Generally, periodicity indicates the ordered recurrence, where some specific patterns are naturally or manually embedded.
Through the lens of the control community, the periodicity has long been considered as one of the central issues in both theory and practice.
One of the most common concerns of periodicity lies in handling periodic disturbances, and recent researches have been dedicated to dealing with varying periods \cite{Tomizuka,Yao}.
Researches search for control strategy which aims at striking a balance between periodic sampled-data and event-triggered control \cite{Heemels,Tabuada}.
There also exist researches finding control strategy or specific conditions which enable stable limit cycles or periodic solutions \cite{Xiao,ManchesterT}.  

As the periodicity endows living organisms with positive functional advantages \cite{Rapp,Marder}, it also prevails in biomimetic locomotions. The famous central pattern generator produces fruitful results in the interaction of neuroscience and robotics, which helps in achieving rhythmic locomotions under complicated situations \cite{Ijspeert}. 
Averaging method is also a very effective approach to develop general feedback control strategies for underactuated dynamical systems, especially for biomimetic locomotors, which are generically underactuated dynamical systems driven by periodic inputs \cite{Vela,Maggia}. 
Beyond these techniques, the combination of the Poincar\'e return map and the hybrid invariance theory excels in finding robust and stabilizing continuous-time feedback laws for walking gaits of bipedal robots \cite{MorrisB,Chevallereau2009,Hamed}.

In this paper, we describe the geometrically periodic system with the differential inclusions, whose concept is borrowed from the averaging analysis over angles given in \cite{Sanders}. The geometrical characteristics are addressed to distinguish from the time-based periodic system.
As advocated in \cite{Farkas}, periodic phenomenon sometimes can be characterized by the phase coordinate, and the system is accordingly called cylindrical. Moreover, the utilization of differential inclusion enables us to describe the \emph{uncertainty} in the periodic systems \cite{AubinDI}, which further makes the analyses and theorems cultivated in this soil more practical and applicable especially for economy, biology, quantum theorem, to name just a few. 
The invariance principle states the convergence of bounded solutions to the invariant sets, which has already been generalized to various systems, such as hybrid systems \cite{Hinvari}, switched systems \cite{Bacciotti}.

It is also reported that a common candidate Lyapunov function with a negative semidefinite generalized time derivative along the trajectories of the subsystems suffices to establish invariance-like results \cite{KamalapurkarA,KamalapurkarB}.
Moreover, we use the contraction theory on Finsler manifold to analyze the convergent behavior between any pair of trajectories or solutions, which unfolds the incremental stability analysis to a more expansive and essential extend.
By differentially analyzing the incremental stability of the system, the contraction theory provides an elegant approach to offer an extraordinary glimpse of the overall system only from every local subtle system properties \cite{WangSlo}.
Since all the solutions of a contractive system converge to a unique steady-state solution, it forgets the differences in initial conditions or past inputs \cite{Lohmiller,Forni}.

One of the important motivations of this paper is to attack the uncertain periodic system control problem, such as biomimetic robotic dynamic locomotions.
\Revised{Indeed, the typical description of periodicity normally enjoys time as its coordinate, especially when the periodicity is implanted by exact exogenous forcing. 
However, the periodicity may also behave autonomously and endogenously, or the system is simply perturbed by uncertain exogenous forcing, 
such that the uncertainty arises, bringing us into the circumstance focused in the present paper, which is more pervasive, especially for organisms and autonomous systems, since we are generally reluctant or incapable to measure time to the accuracy required by dynamical motions, nor arbitrarily impose extensive control efforts against the general trend of system evolution, but rather leverage targeted strategies according to a key observation of the trend.
Moreover, due to the lack of the system exact knowledge and the possible uncertain perturbations, the period time measurement becomes inconsistent and nearly impossible. For example, regarding the pendulum system with varying or uncertain length, it is difficult to find its natural frequency and exhausting to specify its periodic behavior with the time coordinate. 
Generally, the uncertainty emerges and we have to accommodate.
Therefore, we seek for other intrinsic coordinates which can geometrically describe the universal behavior of the system, 
which are collectively called phase coordinates.}
To this end, we further employ graphical stability to describe the periodic patterns on this manifold, and further seek the criteria of their invariance and contraction.

The main contributions are concluded as follows:
\begin{enumerate}
  \item The exact definition of the geometrically periodic systems and their averaging approximations are given. With certain conditions specified, we also obtain the accuracy of this geometrically averaging approximation.
  \item The stabilities of the original system and its geometrically averaged one are studied. The corresponding stability of one of the two systems can be verified when the other is guaranteed asymptotically stable or forward invariant.
  \item Both the single-valued and the set-valued periodic pattern contractions can be assured by developing candidate Finsler-Lyapunov functions, whose derivatives need to be guaranteed smaller than certain negative terms. 
\end{enumerate}


Through these three contributions, we successfully build a basic theoretical framework for analyzing the geometrically periodic systems with differential inclusions.

The remainder of the paper is organized in the following manner. In Section {\uppercase\expandafter{\romannumeral2}}, we give the averaging approximations of the geometrically periodic systems. Then we investigate the convergence between the geometrically periodic system and its averaging approximation in Section {\uppercase\expandafter{\romannumeral3}}. In Section {\uppercase\expandafter{\romannumeral4}}, we perform the contraction analysis of the phase-based system on the Finsler manifold. And in Section {\uppercase\expandafter{\romannumeral5}}, we give a general conclusion and a brief outlook.

\section{Geometric Averaging Approximation}

The nature of some specific systems makes it suitable to perform the averaging process in a geometric way.
These systems can be formulated by the following differential inclusions:
\begin{align}
\label{eq:ProbForm}
\begin{array}{*{20}{c}}
   {\dot {\bm x} \in \varepsilon \bm X\left( {\phi ,{\bm x}} \right),}  \\
   {\dot \phi  \in \Omega \left( {\bm x} \right),}  \\
\end{array}{\rm{     }}\begin{array}{*{20}{c}}
   {{\bm x}\left( t_0 \right) = {{\bm x}_0},}  \\
   {\phi \left( t_0 \right) = {\phi _0},}  \\
\end{array}
\end{align}
where $\bm x \in D \subset {\mathbb R^n}$ represents the available
state vector in the $n$-dimensional Euclidean space ${\mathbb R^n}$, with $D$ being the compact, convex, and path-connected domain of the set-valued map $\bm X$, \emph{i.e.} ${\rm dom} \bm X$, and
$\dot \star = {{{\rm{d}}\star} \mathord{\left/
 {\vphantom {{{\rm{d}}\star} {{\rm{d}}t}}} \right.
 \kern-\nulldelimiterspace} {{\rm{d}}t}}$,
$\phi  \in {\mathcal  S^1}$ is the 2$\pi$-periodic angle
which describes the periodic characteristic of the system, and $\varepsilon>0$ denotes a small parameter.
Moreover,  $t_0$ is the initial time instant,
 $\phi_0$ is the initial phase angle on $1$-dimensional sphere $\mathcal S^1$, and $\bm x_0$ is the initial state.
In the discussion of the periodic orbit, $\phi$ can be viewed as a scalar coordinate traveling along the orbit.
\Revised{There are numerous approaches to choose the scalar $\phi$: it may be a natural generalized coordinate for the standard mechanical system such that it can be viewed as a decomposable 1-dimensional state, a measure of an orientation-preserving evolution in the phase space, or yield from the phase description of specific virtual holonomic or nonholonomic constraints (constraints that only related to generalized coordinates and achieved by feedback control rather than physical constraints) for specific gaits.
In certain circumstances, $\phi$ does not only rely on the instantaneous value of the state $\bm x$, but also on
 its history, such that the states with the same value may be assigned to different $\phi$, due to their different evolution histories.}
Moreover, $\bm X:{ {\mathcal  S^1} } \times {\mathbb R^n} \to \rm{comp}\left( {{\mathbb R^n}} \right)$ is the set-valued map to describe the system, with $\rm{comp}\left( {{\mathbb R^n}} \right)$
being the set of all the nonempty compact subsets of ${{\mathbb R^n}}$.
On this set, the Hausdorff metric between those subsets is defined:
\begin{equation}
\label{eq:HausMetric}
{d_{{H}}}\left( {S_0,S_1} \right) = \max \left\{ {\mathop {\sup }\limits_{\bm x_0 \in  S_0} \underline d\left( {\bm x_0, S_1} \right),\mathop {\sup }\limits_{\bm x_1 \in  S_1} \underline d\left( {\bm x_1,S_0} \right)} \right\},\nonumber
\end{equation}
where $\underline d\left( {\bm x_0, S_1} \right) = \mathop {\inf }_{\bm x_1 \in S_1} d\left( {\bm x_0,\bm x_1} \right)$, $\bm x_0, \bm x_1 \in \mathbb R^n$, $S_0, S_1 \in {\rm{comp}}\left( {{\mathbb R^n}} \right)$ and $d_H:\rm{comp}\left( {{\mathbb R^n}} \right) \times \rm{comp}\left( {{\mathbb R^n}} \right) \to \mathbb R_{\ge 0}$ is a pseudometric on $\rm{comp}\left( {{\mathbb R^n}} \right)$, $\mathbb R_{\ge 0}:=\left[0,+\infty\right)$.
Additionally, $ \Omega: {\mathbb R^n} \to \rm{comp}\left( {{\mathbb R}} \right)$ is another set-valued map for angular velocity, whose range, defined as ${{\rm rge}}~\Omega : = \bigcup_{\bm x \in D} { \Omega\left(\bm x \right)}$, is supposed to be a nonempty compact subset of 1-dimensional Banach space.
The original system~\eqref{eq:ProbForm} can be associated with an averaged system, which is formulated as
\begin{align}
\label{eq:avProbForm}
\begin{array}{*{20}{c}}
   {\dot {\bm y} \in \varepsilon \bar {\bm X} \left( \bm y \right),}  \\
   {\dot \psi  \in \Omega \left( \bm y \right),}  \\
\end{array}\begin{array}{*{20}{c}}
   {\bm y\left( t_0 \right) = {\bm x_0},}  \\
   {\psi \left( t_0 \right) = {\phi _0},}  \\
\end{array}
\end{align}
where the averaged right-hand side is given by
\begin{align}
\label{eq:avrInt}
\bar {\bm X} \left( \bm x \right) = \frac{1}{{2\pi }}\oint_{\mathcal  S^1} \bm X \left( {\phi ,\bm x} \right){\rm{d}}\sigma_{\mathcal S^1},
\end{align}
with $\sigma_{\mathcal S^1}$ being the $\sigma$-finite surface measure of ${\mathcal S^1}$, and the integral of differential inclusion performed in the sense of Aumann \cite{Aumman}, where the measure is non-atomic 
(\emph{c.f.} \cite{Aubin} for more details).
The averaging is analog to the single-valued version discussed in \cite{Sanders}, and its scope is confined to the so-called regular case, requiring the following assumption.

\begin{Assumption}\label{as:Assumption1}
For $\forall \bm x \in D$, $0 {\rm{ < }}m \le \mathop {\inf }_{\omega  \in \Omega \left( \bm x \right)} \left\| {\omega } \right\| \le \mathop {\sup }_{\omega  \in \Omega\left( \bm x \right) } \left\| {\omega } \right\| \le M < \infty $, with $m$ and $M$ being positive constants.
The set-valued map $\Omega \left( \bm x \right)$ is convex-valued and continuous, which implies that either ${\rm{rge}}~\Omega \subset \left[ {m,M} \right]$, or ${\rm{rge}}~\Omega \subset \left[ {-M,-m} \right]$ and the angle $\phi$ is strictly monotonic.
\end{Assumption}
\Revised{
\begin{Remark}
\label{Re:Haus}
The convergence \emph{w.r.t.} Hausdorff metric is not equivalent to the ordinary set convergence without the boundedness restriction \cite{Rockafellar}.
The latter notion is more generalized comparing to the former one.
Thus, since the Hausdorff metric is employed as the metric for the space of nonempty and closed sets in this paper, we actually suppose that the sets equipped with Hausdorff metric in subsequent discussions are bounded.
\end{Remark}}

The most natural solution of the differential inclusion, dedicated to describe physical systems, is probably the absolutely continuous solution \cite{SFilippov}, being differentiable almost everywhere (\emph{a.e.}), whose specific definition is given as follows.

\begin{Definition}[Carath\'eodory trajectory]
\label{def:AC-solution}
Carath\'eodory trajectories generated by the differential inclusion
\begin{equation}
\dot {\bm x} \in \bm X\left( {t,\bm x} \right),\bm x\left( {t_0} \right) = {\bm x_0} \in R_0 \in {\rm{comp}}(\mathbb{R}^n),\nonumber
\end{equation}
is the set-valued map $R:\mathbb{R}_{\ge {t_0}}\to {\rm comp}(\mathbb{R}^n)$, $R({t_0})=R_0$ that collects all the viable absolutely continuous, \emph{i.e.} Lipschitz everywhere and differentiable with Lebesgue-integrable derivatives \emph{a.e.}, single-valued trajectories $\bm x(t)$ which satisfies $\dot {\bm x} \in \bm X\left( {t,\bm x} \right)$ \emph{a.e.} on $t \in [t_0,T]$, and $\bm x\left( t_0 \right) = {\bm x_0}$.
\end{Definition}

\Revised{\begin{Definition}[Measurable \& Locally integrable selection]
\label{def:Measurable-selection}
Considering a set-valued map from $X$ to the power set of $Y$, $F:X \to {\mathcal P}\left( Y \right)$, a function $f:X \to Y$ is a selection of $F$, if $\forall \bm x \in X,f\left(\bm x \right) \in F\left(\bm x \right)$, whose existence can be validated by the axiom of choice. If this function is measurable, \emph{i.e.} the preimage of any open set in $Y$ is measurable, then the selection is called measurable selection. In this work, the measurable selection from the set-valued map solution is \emph{w.r.t.} the time $t$, or the phase angle $\phi$. Furthermore, the measurable selection $f$ is called locally integrable selection, if it is locally integrable, \emph{i.e.} for every compact set $U$ in $X$, $f$ is integrable on $U$.
\end{Definition}}

\Revised{\begin{Definition}[Lipschitz continuity of set-valued maps]
\label{def:LC-diffIn}
A set-valued map $\bm S:{\mathbb R^n}\to{\rm{comp}\left(\mathbb R^m\right)}$ is Lipschitz on $D$, if 
\begin{equation}
\bm S\left( {\bm x_0} \right) \subset \bm S\left( \bm x_1 \right) + \lambda_S d\left({\bm x_0, \bm x_1}\right)\mathbb B,\nonumber
 \end{equation}
 for all $\bm x_0, \bm x_1 \in D$, where $\mathbb B$ denotes the closed unit ball whose center is the origin, and $\lambda_S$ is the corresponding Lipschitz constant. And \emph{w.r.t.} $\bm X:{ {\mathcal  S^1} } \times {\mathbb R^n} \to \rm{comp}\left( {{\mathbb R^n}} \right)$, the set-valued map is called Lipschitz in $\bm x$ if 
 \begin{equation}
\bm X\left( \phi,{\bm x_0} \right) \subset \bm X\left( \phi,\bm x_1 \right) + \lambda  d\left({\bm x_0, \bm x_1}\right) \mathbb B,\nonumber
 \end{equation}
 for all $\bm x_0, \bm x_1 \in D$, for \emph{a.e.} $\phi \in \mathcal S^1$, where $\lambda$ is the corresponding Lipschitz constant.
\end{Definition}}

\Revised{\begin{Definition}[Upper semicontinuity of set-valued maps]
\label{def:USS-diffIn}
A set-valued map $\bm S:{\mathbb R^n}\to{\rm{comp}\left(\mathbb R^m\right)}$ is called upper semicontinuous on $D$, if,
 for $\forall \bm x_0, \bm x_1  \in D$, and $\forall \varepsilon > 0$, there exists $\delta > 0$ satisfying that, provided $d ({\bm x_0, \bm x_1})  < \delta$, we have ${\bm S}\left(\bm x_0 \right) \subset{\bm S}\left( \bm x_1  \right) + \varepsilon \mathbb B$.
\end{Definition}}

\begin{Remark}
\label{Re:CCB}
If $\bm X \left( {\phi ,\bm x} \right)$ is upper semicontinuous in both $t$ and $\bm x$, based on the work of Filippov \cite{Filippov},
and the Arzel\`{a}-Ascoli theorem, the Carath\'eodory trajectory sets of the inclusions \eqref{eq:ProbForm} and its geometric averaging \eqref{eq:avProbForm} exist and possess the compactness, connectedness, and boundary properties\footnote{Please refer to \emph{Corollaries} 4.5 - 4.7 in \cite{SmirnovIntro.} for the definition and the proof of these properties.}.
\end{Remark}

\begin{Theorem}\label{th:Theorem1}
With the aforementioned system \eqref{eq:ProbForm} being on the domain $Q: =  \mathcal S^1 \times D$, and \emph{Assumption~\ref{as:Assumption1}} satisfied, if the following conditions hold:

\Revised{\romannumeral1) the mapping $\bm X \left( {\phi ,\bm x} \right)$ is nonempty, compact,
convex-valued, upper semi-continuous for all $\left( {\phi ,\bm x} \right) \in Q$, \emph{i.e.} satisfying the \emph{continuous time basic conditions} in \cite{Kellett}, or the hypothesis of \emph{Proposition S2} in \cite{Cortes}, and is
bounded with constant $M_{ X}$, Lipschitz continuous in $\bm x$ with constant $\lambda$;}

\romannumeral2) in addition to \emph{Assumption~\ref{as:Assumption1}}, the mapping $\Omega \left( {\bm x} \right)$ satisfies the Lipschitz condition with constant $\lambda_\omega$;

\Revised{\romannumeral3) the compact nonempty path-connected set $D$ contains all the initial conditions, from which the solution starts remains in $D$ \emph{w.r.t.} both the original system \eqref{eq:ProbForm} and the averaged system \eqref{eq:avProbForm}, \emph{i.e.} integrally bounded in $\mathbb R^n$ and forward complete in $D$ (Hereafter, the above three conditions and \emph{Assumption~\ref{as:Assumption1}} are concluded as \emph{Condition A}.);}


\romannumeral4) the metric $d$ equipped in the state space is the euclidean metric, \emph{i.e.} $d\left( {\bm x,\bm y} \right) = \left\| {\bm x - \bm y} \right\|$,

then for any  $L > 0$, there exists a constant ${\varepsilon _0}$, \emph{s.t.}, for $\forall \varepsilon  \in \left( {0,{\varepsilon _0}} \right]$, $\forall \bm x_0 \in D$, and $\forall t \in \left[ {t_0,{t_0+L \mathord{\left/
 {\vphantom {t_0+L \varepsilon }} \right.
 \kern-\nulldelimiterspace} \varepsilon }} \right],$
\begin{align}
\label{eq:Th1}
{d_H}\left( {{\rm{cl}}\mathcal R_ t \left( {\varepsilon \bm X,\Omega ,{\bm x_0}} \right),\mathcal R_ t \left( {\varepsilon \bar {\bm X} ,{\bm x_0}} \right)} \right) \le c \varepsilon,
\end{align}
where $c$ is a positive constant, and the time cross-section
$\mathcal R_ t \left( {\varepsilon \bm X,\Omega ,{\bm x_0}} \right): = \left\{ {\bm x\left( t \right)\left| {\bm x\left( t \right) \in {\mathcal S_{\left[ {t_0,{t_0+L \mathord{\left/
 {\vphantom {L \varepsilon }} \right.
 \kern-\nulldelimiterspace} \varepsilon }} \right]}}\left( {\varepsilon \bm X,\Omega ,{\bm x_0}} \right)} \right.} \right\}$,
with ${\mathcal S_{\left[ {t_0,t_0+{L \mathord{\left/
 {\vphantom {L \varepsilon }} \right.
 \kern-\nulldelimiterspace} \varepsilon }} \right]}}$ being the Carath\'eodory trajectory set of \eqref{eq:ProbForm} on $t \in \left[ {t_0,{t_0+L \mathord{\left/
 {\vphantom {L \varepsilon }} \right.
 \kern-\nulldelimiterspace} \varepsilon }} \right]$
 and ${\rm{cl}}\mathcal R_ t $ denoting the closure of $\mathcal R_ t $.
\end{Theorem}

\begin{proof}
In view of the definition of the geometric averaging provided in \eqref{eq:avrInt}, which is essentially uniform and non-autonomous, then,
with condition \romannumeral1) and condition \romannumeral3) satisfied, utilizing the properties of Aumann's integral, it can be drawn that, the averaged set-valued map $\bar {\bm X}$ is
convex and compact, bounded with constant $M_{ X}$, and Lipschitz in $\bm x$ with constant $\lambda$ in the domain $Q$.
According to \emph{Assumption~\ref{as:Assumption1}}, provided that the system \eqref{eq:ProbForm} is kept in $Q$, the angle $\phi$ uniformly rotates in the same direction with a bounded speed. Due to the symmetrical characteristic in rotating forward and backward, it is fair to suppose ${\rm{rge}}~\Omega \subset \left[ {m,M} \right]$.

In order to prove \emph{Theorem~\ref{th:Theorem1}}, let us prove the following statement first:
\begin{align}
\label{eq:Th1aaa}
\mathcal R_ t \left( {\varepsilon \bar {\bm X} ,{\bm x_0}} \right) \subset {\rm{cl}}\mathcal R_ t \left( {\varepsilon {\bm X},\Omega ,{\bm x_0}} \right) + c\varepsilon \mathbb B,
\end{align}
where $\mathbb B$ denotes the closed unit ball with its center at the origin.

\Revised{The proof of this statement is divided into 4 steps, whose sketch is provided to facilitate further explanation. In the first step, the distance between the solution of the averaged system $\bm y$ and an approximation trajectory $\bm y'$ is considered. And in the second step, the distance between two different approximation trajectories $\bm x'$ and $\bm y'$ is discussed. Further, in the third step, we consider the distance between the solution of the original system $\bm x$ and approximation trajectory $\bm x'$. Finally in the fourth step, we perform the analysis regarding the completeness, and the conclusion is reached by implementing the subadditivity.}

\emph{\textbf{Step 1:}}
For convenience, we accumulate $\phi$ on $\mathbb R_{\ge \phi_0}$, which suggests that $\phi$ can be $0,~2\pi,~4\pi \cdots $ or any other nonnegative value. In the enlightenment of Poincar\'e map, or more probably by the observation that the angle $\phi$ unavoidably increases and consecutively passes every single point on $\mathcal S^1$, we divide the time interval $\left[ {t_0,{t_0+L \mathord{\left/
 {\vphantom {L \varepsilon }} \right.
 \kern-\nulldelimiterspace} \varepsilon }} \right]$ by the time points $t_k$ that $\phi(t_k) {\rm{ = }}{\phi _{\rm{0}}}{\rm{ + 2}}k\pi$, where $k \in \mathbb{Z}_+$, and $\mathbb{Z}_+$ is the set of positive integers. Here, we define the recurrent-time sequence as $R{t}_{\left[ {t_0,{t_0+L \mathord{\left/
 {\vphantom {L \varepsilon }} \right.
 \kern-\nulldelimiterspace} \varepsilon }} \right]}: = \left\{ {\delta {t_1},\delta {t_2} \cdots \delta {t_{{N_r}}}} \right\}$, where ${N_r}\in \mathbb{Z}_+$, for $\forall k < {N_r},{t_k} = \sum\nolimits_1^k {\delta {t_i}}$, $t_{Nr} = {L \mathord{\left/
 {\vphantom {L \varepsilon }} \right.
 \kern-\nulldelimiterspace} \varepsilon }$, and for $\forall t \in \left( {{t_{{N_r} - 1}},{t_{{N_r}}}} \right],$ with $I_{\phi t}$ being the interval between $\phi_0$ and $\phi(t)$ on $\mathcal S^1$, and $I_{\phi tc}$ the complement of $I_{\phi t}$, both ${\sigma _{\mathcal S^1}}\left( {{I_{\phi t}}} \right) > 0$, and ${\sigma _{\mathcal S^1}}\left( {{I_{\phi tc}}} \right) > 0$.
Furthermore, $\mathcal{C}_{\left[ {t_0,{t_0+L \mathord{\left/
 {\vphantom {L \varepsilon }} \right.
 \kern-\nulldelimiterspace} \varepsilon }} \right]}$ is defined as the collection of all the viable recurrent-time sequences \emph{w.r.t.} the system~\eqref{eq:ProbForm} and \emph{Assumption~\ref{as:Assumption1}}.
Consider an arbitrary solution $\bm y(t)$ of system~\eqref{eq:avProbForm} which satisfies condition \romannumeral3).
Withal, we also select an arbitrary recurrent-time sequence $R{t^*}_{\left[ {t_0,{t_0+L \mathord{\left/
 {\vphantom {L \varepsilon }} \right.
 \kern-\nulldelimiterspace} \varepsilon }} \right]}$ from $\mathcal{C}_{\left[ {t_0,{t_0+L \mathord{\left/
 {\vphantom {L \varepsilon }} \right.
 \kern-\nulldelimiterspace} \varepsilon }} \right]}$, whose points correspond to the time separation points $t_k^*$ , with $k \in \mathbb{Z}_+$ and $k < N_r^*$.
Therefore, there exists measurable selection $\bm v_y(t)\in \varepsilon \bar {\bm X} \left( \bm y \left( t \right)\right)$, such that 
\begin{align}
\label{eq:Intv}
&\bm y\left( t \right) = \bm y\left( {t_k^*} \right) + \int_{t_k^*}^t {{\bm v_y}} \left( \tau  \right){\rm{d}}\tau ,~~t \in \left( {t_k^*,t_{k + 1}^*} \right],\nonumber\\
&{\bm y\left( t_0 \right) = {\bm x_0}}.~~~~~~~~~~~~~~~~~~~~~~~~~~~~~~~~~~~~~~~~~~
\end{align}
Accordingly, consider the following function:
\begin{align}
\label{eq:Intv2}
 &\bm y'\left( t \right) = \bm y'\left( {t_k^ * } \right) + \bm v_y' \left( {t_k^ * } \right)\cdot\left( {t - t_k^ * } \right),~~t \in \left( {t_k^*,t_{k + 1}^*} \right], \nonumber\\
 &\bm y'\left( t_0 \right) = {\bm x_0},~~~~~~~~~~~~~~~~~~~~~~~~~~~~~~~~~~~~~~~~~~~~~~~
\end{align}
where $\bm v_y' \left( {t_k^ * } \right) \in \varepsilon \bar{\bm X}\left( {\bm y'\left( {t_k^ * } \right)} \right)$, more specifically,
\begin{align}
\label{eq:argmin}
\bm v_y' \left( {t_k^ * } \right) = \arg \min \{ {d ({\delta t_{k+1}^ * \cdot \bm v , \int_{t_k^ * }^{t_{k + 1}^ * } {\bm v_y \left( \tau  \right){\rm{d}}\tau } } ) }\nonumber\\
\left. {\left| {\bm v \in \varepsilon \bar {\bm X} \left( {\bm y'\left( {t_k^ * } \right)} \right)} \right.} \right\}.
\end{align}
In addition, for the sake of consistency, we define $t_0^* \buildrel \Delta \over = t_0$, even though $t_0^*$ need not be a time separation point.

With these two specific functions \eqref{eq:Intv} and \eqref{eq:Intv2}, we can further investigate the upper bound of the distance between $\bm y\left( {L \mathord{\left/
 {\vphantom {L \varepsilon }} \right.
 \kern-\nulldelimiterspace} \varepsilon } \right)$ and $\bm y'\left( {L \mathord{\left/
 {\vphantom {L \varepsilon }} \right.
 \kern-\nulldelimiterspace} \varepsilon } \right)$ in the metric $d$.

Firstly, when the system is in the first round, that is, $t \in \left( {t_0,t_1^*} \right]$, the distance is given by
\begin{equation}
\label{eq:disdif0}
d\left( {\bm y '\left( t \right),{\bm y }\left( t \right)} \right) = d\left( {t  \cdot \bm v_y'\left( t_0 \right),\int_{t_0}^{t} {{\bm v_y}\left( \tau  \right){\rm{d}}\tau } } \right).\nonumber
\end{equation}
As a coarse estimate, since $\bm v_y(t)\in \varepsilon \bar {\bm X} \left( \bm y \left( t \right)\right)$, we have
\begin{equation}
\label{eq:coarse0}
d\left( {\bm y'\left( {t_0} \right),\bm y\left( t \right)} \right) \le \varepsilon {M_{ X}}\delta t_1^*.\nonumber
\end{equation}
Then, due to the optimized selection of ${{\bm v'}_y}\left( {t_0} \right)$ in the sense of \eqref{eq:argmin}, it is straightforward to see that
\begin{align}
\label{eq:inf0}
 &d\left( {\delta t_1^ *  \cdot {{\bm v'}_y}\left( {t_0} \right),\int_{t_0}^{{t_0}+\delta t_1^ * } {{\bm v_y}\left( \tau  \right){\rm d}\tau } } \right) \nonumber\\
  \le& \varepsilon {d_H}\left( {\int_{t_0}^{{t_0}+\delta t_1^ * } {\bar {\bm X}\left( {\bm y'\left( 0 \right)} \right){\rm d}\tau } ,\int_{t_0}^{{t_0}+\delta t_1^ * } {\bar {\bm X}\left( {\bm y\left( \tau  \right)} \right){\rm d}\tau } } \right) \nonumber\\
  \le& \int_{t_0}^{{t_0}+\delta t_1^ * } \varepsilon {{d_H}\left( {\bar {\bm X}\left( {\bm y'\left( 0 \right)} \right),\bar {\bm X}\left( {\bm y\left( \tau  \right)} \right)} \right){\rm d}\tau }.
\end{align}
Based on the Lipschitz property of $\bar {\bm X}$, it can be known that
\begin{align}
\label{eq:Lip0}
\varepsilon {d_H}\left( {\bar {\bm X}\left( {\bm y\left( t \right)} \right),\bar {\bm X}\left( {\bm y'\left( {t_0} \right)} \right)} \right) \le {\varepsilon ^2}\lambda {M_X}\delta t_1^*.
\end{align}
From \eqref{eq:inf0} and \eqref{eq:Lip0}, it follows that
\begin{align}
\label{eq:Dis0}
d\left( {\bm y'\left( {t_1^ * } \right),\bm y\left( {t_1^ * } \right)} \right) \le {\varepsilon ^2}\lambda {M_X}{\left( {\delta t_1^ * } \right)^2}.
\end{align}

Secondly, considering the situation when the system is in the $(k+1)$th round, where $k \in \mathbb{Z}_+$ and $k < N_r^*$, there exists
\begin{align}
\label{eq:disdifk}
 &d\left( {\bm y'\left( t \right),\bm y\left( t \right)} \right) \le d\left( {\bm y'\left( {t_k^*} \right),\bm y\left( {t_k^*} \right)} \right) \nonumber\\
 &~~~~~~~~~~~~~ + d\left( {\left( {t - t_k^*} \right){{\bm v'}_y}\left( {t_k^*} \right),\int_{t_k^*}^t {{\bm v_y}\left( \tau  \right){\rm{d}}\tau } } \right),\nonumber
\end{align}
within $t \in \left( {t_k^*,t_{k + 1}^*} \right]$. Hereafter, the distance between the two functions at the instant $t_k^*$, $d\left( {\bm y'\left( {t_k^*} \right),\bm y\left( {t_k^*} \right)} \right)$ is abbreviated as $d{y_k}$. The coarse estimate can be obtained:
\begin{equation}
\label{eq:coarsek}
d\left( {\bm y'\left( {t_k^*} \right),\bm y\left( t \right)} \right) \le d{y_k} + \varepsilon {M_X}\delta t_{k+1}^ *.\nonumber
\end{equation}
Then, it is obvious that, by performing the similar process in the analysis of the first round, the following inequality holds:
\begin{align}
\label{eq:Disk}
d{ y_{k+1 }} \le (1+\varepsilon \lambda \delta t_{k+1}^ *) \cdot d{y_{k}}  +  {\varepsilon ^2} \lambda {M_X}{\left( {\delta t_{k+1}^ * } \right)^2},
\end{align}
with $k \in \mathbb{Z}_+$ and $k < N_r^*$.
In view of the supposition that $\rm{rge} (\Omega) \subset \left[ {m,M} \right]$, we have $\delta t_{k+1}^ *  \le {{2\pi } \mathord{\left/
 {\vphantom {{2\pi } m}} \right.
 \kern-\nulldelimiterspace} m},N_r^* \le {{ML} \mathord{\left/
 {\vphantom {{ML} {\left( {2\pi \varepsilon } \right)}}} \right.
 \kern-\nulldelimiterspace} {\left( {2\pi \varepsilon } \right)}}$.
Thus, also by the virtue of \eqref{eq:Dis0} and \eqref{eq:Disk}, it is known that
\begin{align}
\label{eq:DisN}
d{y_{{N^*_r}}} \le \left( {{e^{{{\lambda ML} \mathord{\left/
 {\vphantom {{\lambda ML} m}} \right.
 \kern-\nulldelimiterspace} m}}} - 1} \right){{2\pi {M_X}} \mathord{\left/
 {\vphantom {{2\pi {M_x}} m}} \right.
 \kern-\nulldelimiterspace} m} \cdot \varepsilon.
\end{align}

\emph{\textbf{Step 2:}}
What follows is the part regarding to the original system \eqref{eq:ProbForm}. Consider the following function:
\begin{align}
\label{eq:FunBri2a}
 &\bm x'\left( t \right) = \bm x'\left( t_k^* \right) + \int_{t_k^*}^t {{{\bm v'}_x}\left( \tau  \right)} \rm{d}\tau ,~~t \in \left( {t_k^*,t_{k + 1}^*} \right],  \nonumber\\
 &\bm x'\left( {t_0} \right) = {\bm x_0}.
\end{align}
Correspondingly, the dynamics of $\psi'$ is given by
\begin{align}
\label{eq:FunBri2b}
 &\psi '\left(t\right) = \psi '\left( {t_k^*} \right) + {{v'}_\psi }\left( {t_k^*} \right) \cdot \left( {t - t_k^*} \right),~~t \in \left( {t_k^*,t_{k + 1}^*} \right],  \nonumber\\
 &\psi '\left( {t_0} \right) = {\phi _0},
\end{align}
where ${{v'}_\psi }\left( {t_k^*} \right) = {{2\pi } \mathord{\left/
 {\vphantom {{2\pi } {\delta t_k^*}}} \right.
 \kern-\nulldelimiterspace} {\delta t_{k+1}^*}}$. Then, \emph{w.r.t.} function \eqref{eq:FunBri2a}, we have ${{\bm v'}_x}\left( t \right) \in \bm X\left( {\psi '\left( t \right)}, \bm y'\left( {t_k^*} \right) \right)$, \emph{s.t.}
\begin{equation}
\label{eq:SingleBound}
\int_{t_k^*}^{t_{k + 1}^*} {{{\bm v'_x}}\left( \tau  \right)} {\rm{d}}\tau  = {{\bm v'_y}}\left( {t_k^*} \right) \cdot \delta t_{k + 1}^*,\nonumber
\end{equation}
with $k \in \mathbb{Z}_+$ and $k < N_r^*-1$.
The existence of function ${{\bm v'}_x}$ roots in the definition of Aumann integral of differential inclusions. Obviously, along with uniform rotation speed within each interval $\delta t_k^*$, as formulated in \eqref{eq:FunBri2a}, we have
\begin{align}
\label{eq:aveexist}
 \int_{t_k^*}^{t_{k + 1}^*} {{{\bm v'}_x}\left( \tau  \right)} {\rm{d}}\tau  &= \int_{t_k^*}^{t_{k + 1}^*} {{{\bm v'_x}}\left( {\psi '\left( \tau  \right)} \right),\bm y'\left( {t_k^*} \right)} {\rm{d}}\tau  \nonumber\\
  &= \frac{{\delta t_{k + 1}^*}}{{2\pi }}\int_0^{2\pi } {{{\bm v'_x}}\left( {\varphi, \bm y'\left( {t_k^*} \right)} \right)} {\rm{d}}\varphi \nonumber\\
  &= \delta t_{k + 1}^*\cdot{{\bar {\bm v}'}_x}\left( {\bm y'\left( {t_k^*} \right)} \right).
\end{align}
Since ${{\bar {\bm v}'}_x}\left( {\bm y'\left( {t_k^*} \right)} \right) \in \bar {\bm X}\left( {\bm y'\left( {t_k^*} \right)} \right)$, and according to \emph{Theorem~\ref{th:Theorem1}}, the existence of ${{\bm v'}_x}\left( t \right)$ in $t \in \left( {t_k^*,t_{k + 1}^*} \right]$ is obvious. Furthermore, with the functions given in
\eqref{eq:FunBri2a} and \eqref{eq:FunBri2b}, there exists the following inequality:
\begin{align}
\label{eq:DisXY}
 &d\left( {\bm x'\left( t \right),\bm y'\left( t \right)} \right) \le {{4{M_X}\pi  \cdot \varepsilon } \mathord{\left/
 {\vphantom {{4{M_X}\pi  \cdot \varepsilon } m}} \right.
 \kern-\nulldelimiterspace} m},~~t \in \left( {0,{L \mathord{\left/
 {\vphantom {L \varepsilon }} \right.
 \kern-\nulldelimiterspace} \varepsilon }} \right], \\
 \label{eq:DisXY0}
 &d\left( {\bm x'\left( {t_k^*} \right),\bm y'\left( {t_k^*} \right)} \right) = 0,~~k \in {\mathbb Z_ + },k < {N_r^*}.
\end{align}

\emph{\textbf{Step 3:}}
 Consider a solution of the inclusions \eqref{eq:ProbForm}, which is subjected to the selection  $R{t^*}_{\left[ {{t_0},{{t_0}+L \mathord{\left/
 {\vphantom {L \varepsilon }} \right.
 \kern-\nulldelimiterspace} \varepsilon }} \right]}$ from $\mathcal{C}_{\left[ {{t_0},{{t_0}+L \mathord{\left/
 {\vphantom {L \varepsilon }} \right.
 \kern-\nulldelimiterspace} \varepsilon }} \right]}$, and minimize the distance between $\bm v_x$ and $\bm v'_x$ similar to \eqref{eq:argmin}.
 The specific form is given by
 \begin{align}
\label{eq:Xselct}
{\bm v_x}\left( {\varphi \left( t \right),\bm x\left( t \right)} \right) = &\arg \min \left\{ {d\left( {{{\bm v'}_x}\left( {\psi \left(t_\varphi \right),\bm x'\left( {{t_\varphi }} \right)} \right),\bm v} \right)}\right.\nonumber\\
&\left.{\left| {\bm v \in \bm X\left( {\varphi \left( t \right),\bm x\left( t \right)} \right)} \right.} \right\},
\end{align}
where $t \in \left( {t_k^*,t_{k + 1}^*} \right]$, $k \in \mathbb{Z}_+$ and $k < N_r^*-1$. With the angle synchronized, we have ${t_\varphi }: = \left\{ {\psi \left( \tau  \right) = \varphi \left( t \right)\left| {\tau \in \left( {t_k^*,t_{k + 1}^*} \right]} \right.} \right\}$. Since $\rm{rge} (\Omega) \subset \left[ {m,M} \right]$, the map ${t_\varphi }\left( t \right) \in \left( {t_k^*,t_{k + 1}^*} \right] \to \left( {t_k^*,t_{k + 1}^*} \right]$ is bijective. Similarly, within $t \in \left( {t_k^*,t_{k + 1}^*} \right]$, the selection of the rotational speed is given by
 \begin{align}
\label{eq:phiselct}
{v_\varphi }\left( {\bm x\left( t \right)} \right) = \arg \min \left\{ {\| {v-\frac{{2\pi }}{{\delta t_{k + 1}^ * }}} \|\left| {v \in \Omega \left( {\bm x\left( t \right)} \right)} \right.} \right\}.
 \end{align}
As in the analysis process of the distance between $\bm y$ and $\bm y'$, the coarse estimate in the $(k+1)$th round is formulated as
 \begin{align}
\label{eq:DisXs}
d\left( {\bm x'\left( {t_k^ * } \right),\bm x\left( t \right)} \right) \le d{x_k} + \varepsilon {M_X}\delta t_{k + 1}^ * ,
 \end{align}
 where $t \in \left( {t_k^*,t_{k + 1}^*} \right]$.
 Then, consider the additional distance generated in the $(k+1)$th round, which satisfies
 \begin{align}
 &d\left( {\int_{t_k^ * }^{t_{k + 1}^ * } {{\bm v_x}\left( \tau  \right){\rm{d}}\tau } ,\int_{t_k^ * }^{t_{k + 1}^ * } {{{\bm v'}_x}\left( \tau  \right){\rm{d}}\tau } } \right)
 \nonumber\\
  = &d\left( {\int_0^{2\pi } {{1 \mathord{\left/
 {\vphantom {1 {{v_\varphi }\left( {x\left( {{t_\varphi }} \right)} \right)}}} \right.
 \kern-\nulldelimiterspace} {{v_\varphi }\left( {\bm x\left( {{t_\varphi }} \right)} \right)}} \cdot {\bm v_x}\left( {\varphi ,\bm x\left( {{t_\varphi }} \right)} \right){\rm{d}}\varphi } ,} \right.
 \nonumber\\
 &~~~~~~~~~~~~~~~~~~~~~~\left.{\frac{{\delta t_{k + 1}^ * }}{{2\pi }}\int_0^{2\pi } {{{\bm v'}_x}\left( {\varphi ,\bm y'\left( {t_k^ * } \right)} \right){\rm{d}}\varphi } } \right)
  \nonumber\\
  \le & \frac{{\delta t_{k + 1}^ * }}{{2\pi }}d\left( {\int_0^{2\pi } {{\bm v_x}\left( {\varphi ,\bm x\left( {{t_\varphi }} \right)} \right){\rm{d}}\varphi } ,\int_0^{2\pi } {{{\bm v'}_x}\left( {\varphi ,\bm y'\left( {t_k^ * } \right)} \right){\rm{d}}\varphi } } \right)
 \nonumber\\
  + & d\left( {\int_0^{2\pi } {{1 \mathord{\left/
 {\vphantom {1 {{v_\varphi }\left( {x\left( {{t_\varphi }} \right)} \right)}}} \right.
 \kern-\nulldelimiterspace} {{v_\varphi }\left( {\bm x\left( {{t_\varphi }} \right)} \right)}} \cdot {\bm v_x}\left( {\varphi ,\bm x\left( {{t_\varphi }} \right)} \right){\rm{d}}\varphi } ,} \right.
 \nonumber\\
 &~~~~~~~~~~~~~~~~~~~~~~~\left.{\frac{{\delta t_{k + 1}^ * }}{{2\pi }}\int_0^{2\pi } {{\bm v_x}\left( {\varphi ,\bm x\left( {{t_\varphi }} \right)} \right){\rm{d}}\varphi }} \right).\nonumber
 \end{align}
According to \eqref{eq:DisXY0}, there exists $\bm x'\left( {t_k^ * } \right) = \bm y'\left( {t_k^ * } \right)$. And it is straightforward to see that
 \begin{align}
\begin{split}
 d\left( {\int_0^{2\pi } {{\bm v_x}\left( {\varphi ,\bm x\left( {{t_\varphi }} \right)} \right){\rm{d}}\varphi } ,\int_0^{2\pi } {{{\bm v'}_x}\left( {\varphi ,\bm y'\left( {t_k^ * } \right)} \right){\rm{d}}\varphi } } \right) \nonumber\\
  \le \int_0^{2\pi } {d\left( {{\bm v_x}\left( {\varphi ,\bm x\left( {{t_\varphi }} \right)} \right),{{\bm v'}_x}\left( {\varphi ,\bm x'\left( {t_k^ * } \right)} \right)} \right)} {\rm{d}}\varphi.\nonumber
\end{split}
\end{align}
Therefore, by the virtue of the coarse estimation \eqref{eq:DisXs}, the Lipschitz property of the set-valued map $\bm X \left( {\phi ,\bm x} \right)$, and the optimized measurable selection
${\bm v_x}\left( {\varphi \left( t \right),\bm x\left( t \right)} \right)$, we have
 \begin{align}
 &\int_0^{2\pi } {d\left( {{\bm v_x}\left( {\varphi ,\bm x\left( {{t_\varphi }} \right)} \right),{{\bm v'}_x}\left( {\varphi ,\bm x'\left( {t_k^ * } \right)} \right)} \right)} {\rm{d}}\varphi  \nonumber\\
  \le& \varepsilon \int_0^{2\pi } {{d_H}\left( {\bm X\left( {\varphi ,\bm x\left( {{t_\varphi }} \right)} \right),\bm X\left( {\varphi ,\bm x'\left( {t_k^ * } \right)} \right)} \right)} {\rm{d}}\varphi  \nonumber\\
  \le& 2\pi \lambda \varepsilon \left( {d{x_k} + {M_X}\delta t_{k + 1}^ * \varepsilon } \right).\nonumber
 \end{align}
Additionally, there exists $\int_{t_k^ * }^{t_{k + 1}^ * } {{v_\varphi }\left( {\bm x\left( \tau  \right)} \right)} {\rm{d}}\tau  = 2\pi$ due to the selection of $R{t^*}_{\left[ {{t_0},{{t_0}+L \mathord{\left/
 {\vphantom {L \varepsilon }} \right.
 \kern-\nulldelimiterspace} \varepsilon }} \right]}$. Suppose that ${{2\pi } \mathord{\left/
 {\vphantom {{2\pi } {\delta t_{k + 1}^ * }}} \right.
 \kern-\nulldelimiterspace} {\delta t_{k + 1}^ * }} \notin \Omega \left( {\bm x\left( t \right)} \right)$, for $\forall t \in \left( {t_k^*,t_{k + 1}^*} \right]$. Therefore, due to the selection given in \eqref{eq:phiselct}, as well as the compactness and convexity of $\Omega \left( {\bm x\left( t \right)} \right)$, there exists ${v_\varphi }\left( {\bm x\left( t  \right)} \right) \in \partial \Omega \left( {\bm x\left( t \right)} \right)$, with $\partial \star $ being the boundary of ${\star}$. This is equivalent to say that, either ${v_\varphi }\left( {\bm x\left( t  \right)} \right) = \sup \Omega \left( {\bm x\left( t \right)} \right)$ or ${v_\varphi }\left( {\bm x\left( t  \right)} \right) = \inf \Omega \left( {\bm x\left( t \right)} \right)$. Moreover, within $t \in {D_{\sup }}$, that is when the optimized measurable selection of \eqref{eq:phiselct} is $\sup \Omega \left( {\bm x\left( t \right)}\right)$, $\forall v \in \Omega \left( {\bm x\left( t \right)} \right),v <c_{\phi\sup}< {{2\pi } \mathord{\left/
 {\vphantom {{2\pi } {\delta t_{k + 1}^ * }}} \right.
 \kern-\nulldelimiterspace} {\delta t_{k + 1}^ * }}$, and within $t \in {D_{\inf }}$, $\forall v \in \Omega \left( {\bm x\left( t \right)} \right),v >c_{\phi\inf}> {{2\pi } \mathord{\left/
 {\vphantom {{2\pi } {\delta t_{k + 1}^ * }}} \right.
 \kern-\nulldelimiterspace} {\delta t_{k + 1}^ * }}$.
Therefore, both situations should appear in the time period $\left( {t_k^*,t_{k + 1}^*} \right]$, which implies that there exists at least one discontinuity point in this period.
Suppose that the first discontinuity point appears at $t = t_{dc}$, then there exists arbitrary small $\varepsilon_t>0$ \emph{s.t.} ${d_H}\left( {\Omega \left( {\bm x\left( {{t_{dc}}} \right)} \right),\Omega \left( {\bm x\left( {{t_{dc}} + {\varepsilon _t}} \right)} \right)} \right) = {c_{\phi \inf }} - {c_{\phi \sup }} > 0$.
Furthermore, according to the boundedness of $\bm X \left( {\phi ,\bm x} \right)$, the solution $\bm x$ is Lipschitz continuous. Based on this fact, and also the Lipschitz continuity of $\Omega \left( \bm x \right)$, we know that there exists a contradiction.
Thus, also due to the optimized measurable selection of \eqref{eq:phiselct}, it is known that $\exists  t \in \left( {t_k^*,t_{k + 1}^*} \right]$, ${v_\varphi }\left( {\bm x\left( t \right)} \right)= {{2\pi } \mathord{\left/
 {\vphantom {{2\pi } {\delta t_{k + 1}^ * }}} \right.
 \kern-\nulldelimiterspace} {\delta t_{k + 1}^ * }} $.
Beyond that, similar to the coarse estimate \eqref{eq:DisXs}, it is known that
 \begin{align}
d\left( {\bm x\left( {{t_1}} \right),\bm x\left( {{t_2}} \right)} \right) \le \varepsilon {M_X}\delta t_{k + 1}^ *.\nonumber
\end{align}
Considering the Lipschitz constant $\lambda_\omega$ of $\Omega \left( \bm x \right)$, it can be deduced that
 \begin{align}
\left\| {{v_\varphi }\left( {\bm x\left( \tau  \right)} \right) - \frac{{2\pi }}{{\delta t_{k + 1}^ * }}} \right\| \le \varepsilon {\lambda _\omega }{M_X}\delta t_{k + 1}^ *.\nonumber
\end{align}
Consequently, there exists the following inequality:
 \begin{align}
 &d\left( {\int_0^{2\pi } {{1 \mathord{\left/
 {\vphantom {1 {{v_\psi }\left( {\bm x\left( {{t_\varphi }} \right)} \right)}}} \right.
 \kern-\nulldelimiterspace} {{v_\psi }\left( {\bm x\left( {{t_\varphi }} \right)} \right)}} \cdot {\bm v_x}\left( {\varphi ,\bm x\left( {{t_\varphi }} \right)} \right){\rm{d}}\varphi } ,} \right.
 \nonumber\\
 &~~\left.{\frac{{\delta t_{k + 1}^ * }}{{2\pi }}\int_0^{2\pi } {{\bm v_x}\left( {\varphi ,\bm x\left( {{t_\varphi }} \right)} \right){\rm{d}}\varphi } } \right) \le 2\pi {\lambda _\omega }M_X^2\delta t_{k + 1}^ * {\varepsilon ^2}.\nonumber
\end{align}
According to aforementioned analysis, there exists
\begin{equation}
d{x_{N_r^*\! - \!1}} \!\le\! \left(\! {{e^{\lambda\! M\!L\!/\!m}} \!- \!1} \!\right)\!\left(\! {\frac{{2\pi {M_X}}}{m}\! + \!\frac{{2\pi {\lambda _\omega }M_X^2}}{\lambda }} \!\right)\!\varepsilon \! +\!\mathcal O\!\left( \!\varepsilon ^2 \!\right).\nonumber
\end{equation}
Based on the boundedness of $\bm X \left( {\phi ,\bm x} \right)$, consequently, we have
\begin{align}
\label{eq:DisXNm1}
d{x_{N_r^*}} \le &\left( {{e^{\lambda ML/m}} - 1} \right)\left( {\frac{{2\pi {M_X}}}{m} + \frac{{2\pi {\lambda _\omega }M_X^2}}{\lambda }} \right)\varepsilon  \nonumber \\
&+ {{4{M_X}\pi } \mathord{\left/
 {\vphantom {{4{M_X}\pi } {m \cdot \varepsilon }}} \right.
 \kern-\nulldelimiterspace} {m \cdot \varepsilon }} + \mathcal O\left( {{\varepsilon ^2}} \right).
\end{align}

\emph{\textbf{Step 4:}} So far, the path to the validation of the statement \eqref{eq:Th1aaa} is already blazed, however, this path is based on a specific selection of $R{t^*}_{\left[ {{t_0},{{t_0}+L \mathord{\left/
 {\vphantom {L \varepsilon }} \right.
 \kern-\nulldelimiterspace} \varepsilon }} \right]}$, whose existence is not completely assured yet.
At this juncture, we borrow the covering space concept from the algebraic topology to describe the relationship between the solutions of the inclusions \eqref{eq:ProbForm} and the selections of $R{t^*}_{\left[ {{t_0},{{t_0}+L \mathord{\left/
 {\vphantom {L \varepsilon }} \right.
 \kern-\nulldelimiterspace} \varepsilon }} \right]}$.
And we follow the definition of the covering space given in \cite{Top.}.
A covering space of the topology space $C$ is a topology space $B$ together with a continuous surjective map $p$, \emph{s.t.} $p:B~\to~C$, moreover, every point $c$ of $C$ has a neighbourhood of $U_c$, satisfying that its inverse image $p^{-1}(U_c)$ can be written as the union of disjoint open sets $V_\alpha$, and for each $\alpha$, $V_\alpha$ is mapped homeomorphically onto $U_c$ by $p$. In this circumstance, the so-called continuous surjective map can be acquired by the definition of the $R{t^*}_{\left[ {{t_0},{{t_0}+L \mathord{\left/
 {\vphantom {L \varepsilon }} \right.
 \kern-\nulldelimiterspace} \varepsilon }} \right]}$, thus, ${\mathcal S_{\left[ {{t_0},{{t_0}+L \mathord{\left/
 {\vphantom {L \varepsilon }} \right.
 \kern-\nulldelimiterspace} \varepsilon }} \right]}}\left( {\varepsilon \bm X,\Omega ,{\bm x_0}} \right)$ is a covering space of $\mathcal{C}_{\left[ {{t_0},{{t_0}+L \mathord{\left/
 {\vphantom {L \varepsilon }} \right.
 \kern-\nulldelimiterspace} \varepsilon }} \right]}$.
Consequently, selecting an arbitrary viable recurrent-time sequence $R{t^*}_{\left[ {{t_0},{{t_0}+L \mathord{\left/
 {\vphantom {L \varepsilon }} \right.
 \kern-\nulldelimiterspace} \varepsilon }} \right]}$ indicates that the set of its reverse image in the whole solution set ${\mathcal S_{\left[ {{t_0},{{t_0}+L \mathord{\left/
 {\vphantom {L \varepsilon }} \right.
 \kern-\nulldelimiterspace} \varepsilon }} \right]}}\left( {\varepsilon \bm X,\Omega ,{\bm x_0}} \right)$ is already been regularized in the sense of recurrent time. And the union of these images can cover the whole solution set \emph{w.r.t.} $\left[ {t_0,{t_0+L \mathord{\left/
 {\vphantom {t_0+L \varepsilon }} \right.
 \kern-\nulldelimiterspace} \varepsilon }} \right]$, such that the completeness of the preceding analysis is proved. 

In view of the inequalities \eqref{eq:DisN}, \eqref{eq:DisXY} and \eqref{eq:DisXNm1}, the corresponding constant $c$ in \eqref{eq:Th1aaa} can be chosen as
\begin{align}
 c = & \left( {{e^{\frac{\lambda ML}{m}}} - 1} \right)\left( {\frac{{4\pi {M_X}}}{m} + \frac{{2\pi {\lambda _\omega }M_X^2}}{\lambda }} \right)+\frac{8{M_X}\pi}{m}.\nonumber
\end{align}
\emph{s.t.} the statement
\begin{equation}
\mathcal R_ t \left( {\varepsilon \bar {\bm X} ,{\bm x_0}} \right) \subset {\rm{cl}}\mathcal R_ t \left( {\varepsilon {\bm X},\Omega ,{\bm x_0}} \right) + c\varepsilon \mathbb B\nonumber
\end{equation}
holds true on the time interval $ \left[ {{t_0},{{t_0}+L \mathord{\left/
 {\vphantom {L \varepsilon }} \right.
 \kern-\nulldelimiterspace} \varepsilon }} \right]$.

In a similar vein, for every solution chosen from the solution set ${\mathcal S_{\left[ {{t_0},{{t_0}+L \mathord{\left/
 {\vphantom {L \varepsilon }} \right.
 \kern-\nulldelimiterspace} \varepsilon }} \right]}}\left( {\varepsilon \bm X,\Omega ,{\bm x_0}} \right)$, we can find a corresponding solution
 of the averaged system \eqref{eq:avProbForm} $\mathcal O\left( {{\varepsilon }} \right)$ close, within the time interval $ \left[ {{t_0},{{t_0}+L \mathord{\left/
 {\vphantom {L \varepsilon }} \right.
 \kern-\nulldelimiterspace} \varepsilon }} \right]$, which is basically the inverse version of the preceding analysis.
 And that leads to the dual statement of \eqref{eq:Th1aaa}, formulated as
\begin{align}
\label{eq:Th1bbb}
 {\rm{cl}}\mathcal R_ t \left( {\varepsilon {\bm X},\Omega ,{\bm x_0}} \right)\subset \mathcal R_ t \left( {\varepsilon \bar {\bm X} ,{\bm x_0}} \right) + c\varepsilon \mathbb B.
\end{align}

Additionally, the very fact that, for every solution chosen from the solution set ${\mathcal S_{\left[ {{t_0},{{t_0}+L \mathord{\left/
 {\vphantom {L \varepsilon }} \right.
 \kern-\nulldelimiterspace} \varepsilon }} \right]}}\left( {\varepsilon \bm X,\Omega ,{\bm x_0}} \right)$, there exists a unique corresponding recurrent-time sequence, coincides with the surjection of the covering space.
Finally, the validity of statements \eqref{eq:Th1aaa} and \eqref{eq:Th1bbb} completes this proof.
\end{proof}

The sets corresponding to the optimization operations in \eqref{eq:argmin} and \eqref{eq:Xselct} exist and are compact.
It is obvious that we can develop a strict partial order on $\rm{comp}\left( {{\mathbb R^n}} \right)$, according to the functions minimized respectively in \eqref{eq:argmin} and \eqref{eq:Xselct}.
For optimization operations in \eqref{eq:argmin}, considering the level sets induced by the function being minimized, we can find the minimum value $d_{\rm min}$ of them that its corresponding level set $\Sigma_{\rm min}$ intersects with the compact set $\varepsilon \bar{\bm X}\left( {\bm y'\left( {t_k^ * } \right)} \right)$ by checking along the developed partial order.
Except for the trivial circumstance that $d_{\rm min}=0$, the existence of $d_{\rm min}$ and $\Sigma_{\rm min}$ can be validated via considering the collection of all the level sets that intersect the compact set $\varepsilon \bar{\bm X}\left( {\bm y'\left( {t_k^ * } \right)} \right)$, which is obviously bounded.
By taking the inverted strict partial order, we can use Zorn's Lemma to validate the existence of $d_{\rm min}$ and $\Sigma_{\rm min}$ \cite{Top.}, which boils down to the existence of the optimization. The set $\rm{comp}\left( {{\mathbb R^n}} \right)$ equipped with the metric $d$ is a metric space, where the compact sets are always closed. And the level set here, due to being the boundary of some specific set, is closed.
Therefore, the set $\varepsilon \bar{\bm X}\left( {\bm y'\left( {t_k^ * } \right)} \right) \cap \Sigma_{\rm min}$ is closed, and is compact for being a closed subset of both the compact sets $\varepsilon \bar{\bm X}\left( {\bm y'\left( {t_k^ * } \right)} \right)$ and $\Sigma_{\rm min}$. Similar analyses for the optimization operations in \eqref{eq:Xselct} are omitted here for brevity.


\section{Convergence and Invariance Analysis}
\Revised{
Amongst the subsequent reasonings, the basic idea can be interpreted that, the convergences of the averaged system or the original one can lead to  averaging approximations on infinite intervals. 
Additionally, we also provide the definition and basic observations on the so-called graphical stability to focus on the ``stability of periodic pattern''.}

In order to describe the asymptotic behavior along certain family of orbits holistically, we specify the solution which can be regarded as one category of the R-solution, and generally known as the solution of the integral funnel equation \cite{TolstonogovEq,Panasyuk}.
\Revised{
\begin{Definition}[Absolutely continuous set-valued Map]
A set-valued map $\bm S:\mathbb{R}_{\ge {t_0}} \to {\rm comp}(\mathbb{R}^n)$ is called absolutely continuous, if for every $\varepsilon>0$, there exists a corresponding $\delta(\varepsilon)>0$ \emph{s.t.} for any $n\in\mathbb N$, and any finite collection $\left\{ {{\alpha _i},{\beta _i}} \right\}_{i = 1}^N \in {\mathbb R_{ \ge {t_0}}}$ with ${\alpha _1} < {\beta _1} \le {\alpha _2} < {\beta _2} \le  \cdots  \le {\alpha _n} < {\beta _n}$, the condition $\sum\nolimits_{i = 1}^n {\left( {{\beta _i} - {\alpha _i}} \right)}  < \delta \left( \varepsilon  \right)$ implies $\sum\nolimits_{i = 1}^n {{d_H}\left( {\bm S\left( {{\alpha _i}} \right),\bm S\left( {{\beta _i}} \right)} \right)}  < \varepsilon$.
\end{Definition}
}
 \Revised{
\begin{Definition}[Rf-solution]
\label{def:Rf-solution}
The absolutely continuous set-valued map $\mathcal R^*:\mathbb{R}_{\ge {t_0}} \to {\rm comp}(\mathbb{R}^n)$, with its compact initial section $R_0=\mathcal R^*({t_0})\subset D$, is an Rf-solution generated by
\begin{equation}
\dot {\bm x} \in \bm G\left( {t,\bm x} \right),\bm x\left( {t_0} \right) = {\bm x_0}, \forall {\bm x_0} \in R_0 \in {\rm{comp}}(\mathbb{R}^n),\nonumber
\end{equation}
denoted as ${\mathcal S}^{*}_{\mathbb R_{{\ge {t_0}}}}(\bm X,R_0)$ with its cross-section at $t$ denoted as ${\mathcal R}^{*}_{t}(\bm X,R_0) \subset D \subset {\mathbb R^n}$, if the relationship
\begin{equation}
\mathop {\lim }\limits_{\delta  \to 0^+} \!\frac{1}{\delta }{d_H}\left( \!{\mathcal R^*\left( {t \!+ \!\delta } \right),\!\bigcup\limits_{\scriptstyle \bm x \in {\mathcal R^*}\left(\! t \!\right) \hfill \atop 
  \scriptstyle \bm g \in \bm G\left(\! {t,\bm x}\! \right) \hfill}\! {\left\{\! {\bm x \!+ \!\int_t^{t \!+\! \delta } {\!\bm g\!\left( {s,\bm x(s)} \right){\rm{d}}s} } \!\right\}\!} } \!\right) \!=\! 0\nonumber
\end{equation}
holds \emph{a.e.} on $t \in [{t_0},t_0+\beta) \subset \mathbb{R}_{\ge {t_0}}$, $\beta>0$,
where $\bm g$ is the locally integrable selection in $\bm G\left( {t,\bm x} \right)$ on $\left( {t,t + \delta} \right] \times D$,
and $\delta \in \mathbb{R}_+$, $\mathbb{R}_+:=\left(0,+\infty\right)$, and the state $\bm x(s)$ can be calculated by ${\bm x \!+ \!\int_t^{t \!+\! s } {\!\bm g\!\left( {\tau,\bm x(\tau)} \right){\rm{d}}\tau} }$, which is well-defined on every $(t,t+\delta]$. Even though the aforementioned definition of Rf-solution \emph{w.r.t.} $\bm G\left( {t,\bm x} \right)$ follows the time-varying form, the corresponding definition for time invariant system is also included, since the set-valued field merely remains unchanged.
\end{Definition}}

\begin{Remark}
\label{Re:R-solutionp} 
Bearing considerable resemblance to the notion of the sample-and-hold solution in \cite{Clarke}, this solution notion also excludes closed convex hull operation without enduring a nonzero-Lebesgue-measure interval of evolution, which may be generated by the famous Filippov regularization:
\begin{equation}
F\!\left( \bm G \right)\left( {t,\bm x} \right): = \!\bigcap\limits_{\delta  > 0}\! {\bigcap\limits_{\mu \left( S \right) = 0}\! {{\overline{ \rm{co}}}\left\{ {\bm G\left( \!{t,\bm y} \right)\left| {\bm y \in (\bm x+ \delta \mathbb B) \backslash S} \right.} \right\}} },\nonumber
\end{equation}
with ${\overline{ \rm{co}}}(\star)$ denoting the closed convex hull of $\star$, and $\mu$ being a Lebesgue measure, the set $S$ denotes the set where the set-valued map $\bm G\left( \!{t,\bm x} \right)$ is not essential, \emph{viz.} the zero-Lebesgue-measure set of points not defined, or discontinuous.
\end{Remark}

\begin{Lemma}[Existence of the Rf-solution]
\label{Le:Exists} 
With \emph{Condition A} satisfied, taking any compact nonempty set $R_0 \in {\rm{comp}}(D)$ as the initial section, and for $\forall t \in \mathbb R _{\ge t_0}$, there exists forward complete Rf-solutions ${\mathcal S^*_{\mathbb R_{{\ge {t_0}}}}}\left( {\varepsilon \bm X,\Omega ,{R_0}} \right)$ and ${\mathcal S^*_{\mathbb R_{{\ge {t_0}}}}}\left(\varepsilon \bar {\bm X},{R_0}\right)$ defined on $\mathbb R _{\ge t_0} \times D$, for the original system $ \bm X\left( {\phi \left(t\right),{\bm x\left(t\right)}} \right)$ and the averaged system $\bar {\bm X} \left( \bm x\left(t\right) \right)$.
\end{Lemma}

\begin{proof} 
This lemma is analog to the existence assertion in \emph{Remark~\ref{Re:CCB}}. One can then follow the pathway of \emph{Corollary~4.4} in \cite{SmirnovIntro.}.
First of all, the existence of absolutely continuous solutions immediately indicates the existence of locally integrable functions \cite{AsympTeel}.
In other word, the solutions are integrals of their own derivatives, which admits local solutions.
The obstacle confronted is that we need other ways to construct the ``path of sets''. By \emph{Definition~\ref{def:Rf-solution}}, Rf-solution is locally ``set of paths'' in every partition $\delta$, then we can construct and shrink the partitions $\delta$ of every absolutely continuous single-valued function over finite or infinite time interval, which is similar to the construction of the sample-and-hold solutions in \cite{Clarke}.

Every limit of the uniformly convergent subsequence, due to the Arzel\`{a}-Ascoli theorem, is an absolutely continuous solution. Thus, the existence of these single-valued solutions leads to the existence of their collection, the Rf-solution.

Benefit from the forward invariance of $D$ and the definition of $\bm X$ and $\bm {\overline X}$ on $D$, also as a consequence of the fact that the locally integrable selection always exists for $\forall \bm x \in D$ on every partition $\delta$, the solutions are forward complete.
\end{proof}

Although the main difference between the classical dynamical system framework and the differential inclusion
one, is the lack of uniqueness of solutions \cite{Bacciotti}, we can still find the uniqueness for the concept of ``path of set'', the Rf-solution.
\begin{Lemma}[Uniqueness of the Rf-solution]
\label{Le:Uniq}With \emph{Condition A} satisfied,
further, taking any compact nonempty set $R_0 \in {\rm{comp}}(D)$ as the initial section, and for $\forall t \in \mathbb R _{\ge t_0}$, the forward complete Rf-solutions ${\mathcal S^*_{\mathbb R_{{\ge {t_0}}}}}\left( {\varepsilon \bm X,\Omega ,{R_0}} \right)$ and ${\mathcal S^*_{\mathbb R_{{\ge {t_0}}}}}\left(\varepsilon \bar {\bm X},{R_0}\right)$ defined on $\mathbb R _{\ge t_0} \times D$ are unique and depend continuously on the initial set $R_0$.
\end{Lemma}

\begin{proof}
Consider two initial conditions $R_0$ and $R_1$, or more restrictively $\bm x_0$ and $\bm x_1$, and the Rf-solution cross-section after $\delta \in \mathbb R_+$ length of evolution $\mathcal R_ {t+\delta} \left( {\varepsilon \bm X,\Omega ,{R^*}} \right)$
with $R^*$ being selected from ${\rm{comp}}(D)$, abbreviated as $\mathcal R^0_ {t+\delta}$ and $\mathcal R^1_ {t+\delta}$. And the cross-section of the angle $\phi$ at $t$ is denoted as $\Phi _t^0$ with initial angle $\phi_0$, and the sum of Hausdorff distances ${d_H}\left( {\mathcal R_{t + \delta }^0,\mathcal R_{t + \delta }^1} \right) + {d_H}\left( {\Phi _{t + \delta }^0,\Phi _{t + \delta }^1} \right)$ is abbreviated as ${D_H}\left( {t + \delta } \right)$. Based on \emph{Condition A}, there exists
\begin{align}
 {D_H}\left( {t + \delta } \right) \le& {D_H}\left( t \right){e^{\left( {\lambda  + {\lambda _\omega }} \right)\delta }} \nonumber\\ 
  &+ \int_0^\delta  {{e^{\left( {\lambda  + {\lambda _\omega }} \right)\left( {\delta  - s} \right)}}\left( {{M_X} + M} \right){D_H}\left( t \right){\rm d}s}.  \nonumber
\end{align}
When ${D_H}\left( t_0 \right) = 0$, we have ${D_H}\left( t \right) = 0, \forall t \ge t_0$.
Because ${d_H}\left( {\Phi _{t_0}^0,\Phi _{t_0}^1} \right) =0$ ($\Phi _{t_0}^0 = \Phi _{t_0}^1 = \phi_0$), it is obvious that ${d_H}\left( {\mathcal R_{t_0}^0,\mathcal R_{t_0 }^1} \right)=0$ leads to ${d_H}\left( {\mathcal R_{t}^0,\mathcal R_{t}^1} \right)=0$ for $t \ge t_0$, which verifies the uniqueness. The continuity can be verified by the boundedness of ${e^{\left( {\lambda  + {\lambda _\omega }} \right)\delta }} 
  + \int_0^\delta  {{e^{\left( {\lambda  + {\lambda _\omega }} \right)\left( {\delta  - s} \right)}}\left( {{M_X} + M} \right)ds}$. Similar proof can be applied to ${\mathcal S^*_{\mathbb R_{{\ge {t_0}}}}}\left(\varepsilon \bar {\bm X},{R_0}\right)$ without analysis on the phase angle.
\end{proof}

Based on \emph{Lemma~\ref{Le:Uniq}}, we can determine an Rf-solution by specifying its initial condition, which is frequently used in subsequent analyses. Taking fluid flow as a metaphor of the Rf-solution, its uniqueness concentrates on its holistic morphology, regardless of whether it is turbulent or laminar.

The following lemma is to demonstrate the equivalence between the ``path of set'' (Rf-solution) and the ``set of paths'' (Carath\'eodory trajectory sets, or solution funnel)
in the discussion of, but not limited to, the averaging approximation accuracy.
This equivalence of them relies heavily on the convexity of $\bm X \left( {\phi ,\bm x} \right)$. 
\Revised{\begin{Lemma}
\label{Le:Still}
With \emph{Condition A} satisfied,
then for any  $L > 0$, there exists a constant ${\varepsilon _0}$, \emph{s.t.}, for $\forall \varepsilon  \in \left( {0,{\varepsilon _0}} \right]$, $\forall \bm x_0 \in D$, and $\forall t \in \left[ {t_0,{t_0+L \mathord{\left/
 {\vphantom {t_0+L \varepsilon }} \right.
 \kern-\nulldelimiterspace} \varepsilon }} \right],$
\begin{align}
{d_H}\!\left( {\rm{cl}}{\mathcal R^*_t\left( {\varepsilon \bm X,\Omega ,{\bm x_0}} \right),\mathcal R^*_{  t  }\left( \! {\varepsilon \bar {\bm X} ,{\bm x_0}}\!  \right)}  \right)\! \le \!  c\varepsilon,\nonumber
\end{align}
where the time cross-section
$\mathcal R^*_t \left( {\varepsilon \bm X,\Omega ,{\bm x_0}} \right): = \left\{ {\bm x\left( t \right)\left| {\bm x\left( t \right) \in {\mathcal S^*_{\left[ {t_0,{t_0+L \mathord{\left/
 {\vphantom {t_0+L \varepsilon }} \right.
 \kern-\nulldelimiterspace} \varepsilon }} \right]}}\left( {\varepsilon \bm X,\Omega ,{\bm x_0}} \right)} \right.} \right\}$,
with ${\mathcal S^*_{\left[ {t_0,{t_0+L \mathord{\left/
 {\vphantom {t_0+L \varepsilon }} \right.
 \kern-\nulldelimiterspace} \varepsilon }} \right]}}$ being the Rf-solution of \eqref{eq:ProbForm} on $t \in \left[ {t_0,{t_0+L \mathord{\left/
 {\vphantom {t_0+L \varepsilon }} \right.
 \kern-\nulldelimiterspace} \varepsilon }} \right]$, and $c $ is a positive constant. The Rf-solution cross-section $\mathcal R^*_{  t  }\left( \! {\varepsilon \bar {\bm X} ,{\bm x_0}}\!  \right)$ \emph{w.r.t.} the averaged system \eqref{eq:avProbForm} can be similarly defined.
\end{Lemma}}
\begin{proof} 
\Revised{Due to the compactness of the set of Carath\'eodory trajectories, which is mentioned in \emph{Remark \ref{Re:CCB}}, and the boundedness of the solutions, the sets ${\mathcal S^\ast_{\mathbb R_{{\ge {t_0}}}}}\left( {\varepsilon \bm X,\Omega ,{\bm x_0}} \right)$ and ${\mathcal S^\ast_{\mathbb R_{{\ge {t_0}}}}}\left(\varepsilon \bar {\bm X},{\bm x_0}\right)$ can be interpreted as absolutely continuous set-valued maps.
 }

 \Revised{Select a Carath\'eodory trajectory in ${\mathcal S_{\mathbb R_{{\ge {t_0}}}}}\left( {\varepsilon \bm X,\Omega ,{\bm x_0}} \right)$, and from the Lipschitz continuity of $\Omega \left( {\bm x} \right)$, it can be deduced that the angle solution $\phi\left(t\right)$ is absolutely continuous \emph{a.e.}, and monotonously increasing if we accumulate $\phi$ on $\mathbb R_{\ge \phi_0}$.
 It is straightforward to check that $\bm X\left( {\cdot ,{\bm x}} \right)$ is measurable for $\forall {\bm x} \in D$, and $\bm X\left( {\phi ,\cdot} \right)$ is continuous for $\forall \phi \in {\mathcal S}^1$.
In virtue of the Scorza-Dragoni theorem, for every $\varepsilon_{sd}$ there exists a compact set $K\subseteq {\mathcal S}^1$ \emph{s.t.}  $\mu\left( {\mathcal S}^1\backslash K\right) < \varepsilon_{sd}$, where the restriction of $\bm X\left( {\cdot ,\cdot} \right)$ is continuous.
Taking the time measure and the angle measure as two different $\sigma$-finite measures of the system evolution, the later of which is absolutely continuous \emph{w.r.t.} the former, and due to the Radon–Nikodym theorem, the derivative ${{{\rm d}\phi } \mathord{\left/
 {\vphantom {{{\rm d}\phi } {{\rm d}t}}} \right.
 \kern-\nulldelimiterspace} {{\rm d}t}}$ exists on every measurable set of the evolution, which further leads to that $\bm X\left( {\phi(t),{\bm x(t)}} \right)$ is integrable \emph{w.r.t.} $t$ on the compact set $\left[ {t_0,{t_0+L \mathord{\left/
 {\vphantom {t_0+L \varepsilon }} \right.
 \kern-\nulldelimiterspace} \varepsilon }} \right] \times D$.
Then following the rationale in the proof of \emph{Theorem 4.2.1} of \cite{BanachTol}, we can straightforwardly check the limitation equality in \emph{Definition \ref{def:Rf-solution}} for each solution in ${\mathcal S_{\mathbb R_{{\ge {t_0}}}}}\left( {\varepsilon \bm X,\Omega ,{\bm x_0}} \right)$, by removing specific meager sets violating the right hand density, \emph{s.t.} the equivalence is ensured almost everywhere. Similar analysis can be performed on ${\mathcal S_{\mathbb R_{{\ge {t_0}}}}}\left(\varepsilon \bar {\bm X},{\bm x_0}\right)$.
 }
 \end{proof}

In this paper, 
comparing to tracking a specific trajectory, we put particular emphasis on the pattern how the system evolves and the method to retain a specific pattern.
To elaborate this point, let us consider the following two stability definition of periodic solutions given in \cite{Nonsys}.
On the one hand, the orbital stability can be captured by a closed forward invariant set $\mathcal M$ of the autonomous system $\dot {\bm x} = f({\bm x})$. After define an $\varepsilon$-neighbourhood of $\mathcal M$ by ${U_\varepsilon } = \left\{ {{\bm x} \in {\mathbb R^n}|d\left( {{\bm x},\mathcal M} \right) < \varepsilon } \right\}$, one can further define the stability \emph{w.r.t.} $\mathcal M$ as, for each $\varepsilon>0$, there is $\delta>0$ \emph{s.t.}
${\bm x}\left( {{t_0}} \right) \in {U_\delta } \Rightarrow {\bm x}\left( t \right) \in {U_\varepsilon },\forall t > {t_0}$. On the other hand, the stability can also be defined with the ``$\delta$-$\varepsilon$'' language along a specific solution $\bm x_d(t)$ satisfying $ {\dot{\bm x}_d} = f({\bm x_d})$ with a metric of $d\left( {{\bm x},\bm x_d(t)} \right)$. Then, let us further scrutinize these two definitions.
For the first definition, there are various behaviors in the forward invariant sets, far from the simple periodic orbits.
Even only for the planar systems, the Poincar\'e-Bendixson theorem needs to exclude the critical points to guarantee the existence of a periodic orbit, which considerably deteriorates especially in differential inclusion system.
One can also find other complicated behaviors within the forward invariant sets from the work in \cite{PB1,PB2}.
For the second definition, the non-asymptotic-stability problem arises, which is exhibited in the following example.
Suppose a periodic orbit as $\bm x_d\left( t \right) = {\left[ {\begin{array}{*{20}{c}}
   {\sin t} & { - \cos t}  \\
\end{array}} \right]^\top}$, while the system with specific initial conditions converges to the periodic solution $\bm x_d^*\left( t \right) = {\left[ {\begin{array}{*{20}{c}}
   {\sin \left( {t + \tau } \right)} & { - \cos \left( {t + \tau } \right)}  \\
\end{array}} \right]^\top}$, where $\tau$ is bounded away from zero or multiples of $2\pi$.
Thus, following the second definition, the system is obvious not stable \emph{w.r.t.} $\bm x_d\left( t \right) = {\left[ {\begin{array}{*{20}{c}}
   {\sin t} & { - \cos t}  \\
\end{array}} \right]^\top}$.
However, these two orbits follow the same pattern, such that they all move along the unit circle with a uniform speed, wherein the orbits should be regarded as equivalent to each other in the sense of ``pattern stability''. In other word, it is the offset in time that makes they never converge, but not the pattern.
These defects of the two definitions stimulate us to give a suitable definition of the stability. Inspired by the work in \cite{PeriodicS}, the specific definitions are given as graphical translation-independent stability.

\Revised{\begin{Definition}[Set asymptotic convergence \cite{Hinvari}]
\label{def:Setasymp}
The set $\mathcal M$ together with a $\rho$-neighbourhood is defined on $D$. This set is called stable if for every Rf-solution ${\mathcal S}^{*}_{\mathbb R_{{\ge {t_0}}}}(\bm X,R_0)$, with $R_0 \subset D$ being the compact initial section,
and for every positive constant $\varepsilon>0$, there exists a positive constant $\delta>0$, \emph{s.t.} ${d_H}\left( {{R_0},\mathcal M} \right) \le \delta  \Rightarrow {d_H}\left( {\mathcal R_t^ * \left( {\bm X,{R_0}} \right),\mathcal M} \right) \le \varepsilon$ holds for $t \in \mathbb R_{\ge {t_0}}$.
The set $\mathcal M$ is called attractive if there exists a positive constant $\rho > 0$, \emph{s.t.} every Rf-solution initial section $R_0 \subset \mathcal M + \rho \mathbb B$, remains in $D$ for $t \in \mathbb R_{\ge {t_0}}$, and further satisfies that $\mathop {\lim }\limits_{{\rm{t}} \to \infty } {d_H}\left(\mathcal R_t^ * \left( {\bm X,{R_0}} \right),\mathcal M   \right)= 0$. The Rf-solution time cross-section $\mathcal R_t^ * \left( {\bm X,{R_0}} \right)$ asymptotically converges to $\mathcal M$, if $\mathcal M$ is stable and attractive. 
\end{Definition}}
\Revised{
\begin{Definition}[Graphical translation-independent stability, GTIS] Consider a nominal Rf-solution ${\mathcal S}^{*}_{\mathbb R_{{\ge {t_0}}}}(\bm X,R_0^s)$ with its cross-section ${\mathcal R}^{*}_{t}(\bm X,R_0^s) \subset D \subset {\mathbb R^n}$ of the differential inclusion $\dot {\bm x} \in \bm X\left( {t,\bm x} \right)$, and the cross-section at $t_0$ is a compact nonempty set $R_0^s  \subset D \subset {\mathbb R^n}$. 
The differential inclusion is called GTIS \emph{w.r.t.} ${\mathcal S}^{*}_{\mathbb R_{{\ge {t_0}}}}(\bm X,R_0^s)$,
if for every Rf-solution of the inclusion, namely ${\mathcal S}^{*}_{\mathbb R_{{\ge {t_0}}}}(\bm X,R_0)$, with its cross-section at $t_0$ being an arbitrary compact nonempty set $R_0 \subset D \subset {\mathbb R^n}$,
and for any positive constant $\varepsilon_g \in \mathbb R_+$,
there exist positive constant $\delta$, and a Lipschitz continuous time-translation function $T(t):\mathbb R_{\ge t_0} \to \mathbb R, 
 \left| {T\left( t \right) - T\left( {t'} \right)} \right| \le {\lambda _{KT}}\varepsilon_g\left| {t - t'} \right|, {\lambda _{KT}}\varepsilon_g\in \left[ {0,1} \right)$, for $\forall t, t' > \!t_0$,  
\emph{s.t.}, if we have ${d_H}\!\left(R_0, R_0^s\right) \!\le \!\delta$, then ${d_H}(\! {{\rm{gph}}\mathcal S_{t \!+ \!\varepsilon_g \mathbb B}^ * \left( {\bm X,\!{R_0}} \!\right),{\rm{gph}}\mathcal S_{t \!+ \!T\left(\! t \!\right) \!+ \!\varepsilon_g \mathbb B}^ * \left(\! {\bm X,\!R_0^s} \!\right)} ) \!\le \!\varepsilon_g$, $\forall t > \!t_0$. The Hausdorff metric here is a natural extension of the Hausdorff metric on ${\rm comp} (D)$,
where the metric of time can be naturally induced by its norm, and ${\rm{gph}}\mathcal S_{t \!+ \!\varepsilon_g \mathbb B}^ * \left( {\bm X,\!{R_0}} \!\right)$ is called the $\varepsilon$-graph of the Rf-solution $\mathcal S_{t}^ * \left( {\bm X,\!{R_0}} \!\right)$.
Additionally, the abuse of the symbol $\rm{gph}$ here dedicates to eliminate the order relationship in the Rf-solution set.
\label{def:GStability}
\end{Definition}
\begin{Definition}[Asymptotic graphical translation-independent stability, AGTIS] Likewise, a differential inclusion system $\dot {\bm x} \in \bm X\left( {t,\bm x} \right)$ is called AGTIS \emph{w.r.t.} the nominal solution ${\mathcal S}^{*}_{\mathbb R_{{\ge {t_0}}}}(\bm X,R_0^s)$, if it is GTIS \emph{w.r.t.} ${\mathcal S}^{*}_{\mathbb R_{{\ge {t_0}}}}(\bm X,\!R_0^s)$, and for every Rf-solution of the inclusion, namely ${\mathcal S}^{*}_{\mathbb R_{{\ge {t_0}}}}(\bm X,\!R_0)$, with its cross-section at $t_0$ being an arbitrary compact nonempty set $R_0 \subset D \subset {\mathbb R^n}$,
there exists a Lipschitz continuous time-translation function $T(t):\mathbb R_{\ge t_0} \to \mathbb R, 
\left| {T\left( t \right) - T\left( {t'} \right)} \right| \le {\lambda _{KT}}\eta_g(t)\left| {t - t'} \right|, {\lambda _{KT}}\eta_g(t)\in \left[ {0,1} \right)$, for $\forall t, t' > \!t_0$, 
\emph{s.t.} $\mathop {\lim }\limits_{t \to \infty}  {d_H}( {{\rm{gph}}\mathcal S_{t \!+ \!\eta_g(t) \mathbb B}^ * \left( {\bm X,\!{R_0}} \!\right), {\rm{gph}}\mathcal S_{t \!+\! T\left( t \right) \!+ \!\eta_g(t) \mathbb B}^ * \left(\! {\bm X,\!R_0^s} \!\right)} )=0$, $\forall t > \!t_0$, where the positive Lipschitz continuous function $\eta_g(t)$ asymptotically converges to $0$ as $t \to \infty$.
\label{def:Asymptotic Gstability}
\end{Definition}}
\begin{Remark}
\label{Re:SuitGraph}
The graphical translation-independent stability we discuss here differs in formations from the graphical convergence in \cite{Rockafellar,Graph}, since we additionally endow it with a translation function to modulate the speed along the phase portrait. Consequently, phase synchronizations are then insignificantly correlated with ``pattern stability'', which creates differences in treating delay-aroused errors, especially in periodic systems. The above methodology to describe ``pattern stability'' can apply \emph{mutatis mutandis} to the exponentially graphical stability and the boundedness \emph{w.r.t.} the graphical translation-independent stability. 
\end{Remark}
\Revised{\begin{Theorem}[Critical point convergence]\label{th:Theorem2}With \emph{Condition A} satisfied, if the geometrically averaged system \eqref{eq:avProbForm} is AGTIS \emph{w.r.t.} 
a critical point $\bm y_{sc} \in D \subset {\mathbb R^n}$ as $t \to + \infty$, then there exists ${\varepsilon _0}>0$, \emph{s.t.} for $\forall \varepsilon  \in \left( {0,{\varepsilon _0}} \right],$ $\forall \bm x_0 \in D$, and $\forall t \in \mathbb R_{{\ge {t_0}}},$
\begin{align}
\label{eq:Th2}
{d_H}\!\left( {{\rm{cl}}\mathcal R^*_t\left( {\varepsilon \bm X,\Omega ,{\bm x_0}} \right),\mathcal R^*_{  t  }\left( \! {\varepsilon \bar {\bm X} ,{\bm x_0}}\!  \right)}  \right)\! \le \!  c_{gc} \varepsilon,
\end{align}
where $c_{gc}$ is a positive constant, and the time cross-section
$\mathcal R^*_t \left( {\varepsilon \bm X,\Omega ,{\bm x_0}} \right): = \left\{ {\bm x\left( t \right)\left| {\bm x\left( t \right) \in {\mathcal S^*_{\mathbb R_{{\ge {t_0}}}}}\left( {\varepsilon \bm X,\Omega ,{\bm x_0}} \right)} \right.} \right\}$,
with ${\mathcal S^*_{\mathbb R_{{\ge {t_0}}}}}$ being the Rf-solution of \eqref{eq:ProbForm} on $t \in \mathbb R_{{\ge {t_0}}}$.
\end{Theorem}}

\begin{proof} 
Since the sum of the original time and the time translational function $t \!+\! T\left( t \right)$ in \emph{Definition \ref{def:GStability} and Definition  \ref{def:Asymptotic Gstability} } is strictly monotonically increasing, it is straightforward to see that, 
the geometrically averaged system \eqref{eq:avProbForm} also asymptotically converges to the critical point $\bm y_{sc} \in D \subset {\mathbb R^n}$. 

Define ${T_k} = t_0+{k {{L}}  \mathord{\left/
 {\vphantom {{kL_k} \varepsilon }} \right.
 \kern-\nulldelimiterspace} \varepsilon },k \in {\mathbb Z_+ }$, with $L \in \mathbb R _+$ being a positive constant, and for a specific instant $t$, the corresponding distance in \eqref{eq:Th2} is abbreviated as $d_H(t)$.  According to \emph{Lemma~\ref{Le:Still}}, the inequality \eqref{eq:Th1} still holds for  $\forall t \in \left[ {t_0,{t_0+L \mathord{\left/
 {\vphantom {L \varepsilon }} \right.
 \kern-\nulldelimiterspace} \varepsilon }} \right],$ \emph{w.r.t.} the Rf-solution.
 Therefore, when $k=1$, it is straightforward to see that $d_H(t)<c \varepsilon, \forall t \in (T_0,T_1] $, where $c$ is a positive constant.
 Suppose that  $d_H(T_k)<c_k \varepsilon$, then let us consider the Hausdorff metric between ${{\rm{cl}} \mathcal R_{{T_{k + 1}}}^*\left( {\varepsilon \bm X,\Omega ,{\bm x_0}} \right)}$ and $ {\mathcal R_{{{L\! \mathord{\left/\! 
 {\vphantom {L\!  \varepsilon }} \right.
 \kern-\nulldelimiterspace} \varepsilon }}}^*\left( \! {\varepsilon \bar {\bm X},\! {\rm{cl}}\mathcal R_{{T_k}}^*\left( {\varepsilon   {\bm X},\! \Omega,\! {\bm x_0}} \! \right)} \right)}$, which is denoted as $d_{H1}(T_{k+1})$. By the definition of the time cross-section, it can be drawn that
\begin{align}
{{\rm{cl}} \mathcal R_{{\!T_{k \!+\! 1}}}^*\!\left( \!{\varepsilon\! \bm X,\!\Omega,\!{\bm x_0}} \!\right)}\!  =\! {\rm{cl}}\mathcal R_{{L\! \mathord{\left/
 {\vphantom {L \varepsilon }} \right.
 \kern-\nulldelimiterspace} \varepsilon }}^*\!\left( {\varepsilon \bm X,\!\Omega ,\!{\rm{cl}}\mathcal R_{{T_k}}^*\!\left(\! {\varepsilon \bm X,\Omega ,{\bm x_0}} \!\right)} \right).\nonumber
\end{align}
Moreover, we also have
\begin{align}
\label{eq:unionA}
 {\rm{cl}}\mathcal R_{{L \mathord{\left/
 {\vphantom {L \varepsilon }} \right.
 \kern-\nulldelimiterspace} \varepsilon }}^*\left(\! {\varepsilon \bm X,\!\Omega,\!{\rm{cl}}\mathcal R_k^*} \!\right) \!= \!\bigcup\limits_{{\bm x_t} \!\in {\rm{cl}}\!\mathcal R_k^*} {{\rm{cl}} \mathcal R_{{L \mathord{\left/
 {\vphantom {L_{k+1} \varepsilon }} \right.
 \kern-\nulldelimiterspace} \varepsilon }}^*\left(\! {\varepsilon \bm X,\!\Omega,\!{\bm x_t}} \!\right)},
\end{align}
meanwhile, the similar relation also exist regarding to the set $ {\mathcal R_{{{L \mathord{\left/
 {\vphantom {L \varepsilon }} \right.
 \kern-\nulldelimiterspace} \varepsilon }} }^*\left( {\varepsilon \bar {\bm X},{\rm{cl}}\mathcal R_{{T_k}}^*\left( {\varepsilon   {\bm X},\Omega ,{\bm x_0}} \right)} \right)}$. And, after checking the solution existence according to \emph{Lemma~\ref{Le:Exists}}, it follows that, for $\forall {{\bm x_t} \in {{\rm{cl}} \mathcal R_{{T_{k}}}^*\left( {\varepsilon \bm X,\Omega ,{\bm x_0}} \right)}}$,
\begin{align}
\label{eq:pointDis}
{d_H}\left( {{\rm{cl}}\mathcal R_{{L \mathord{\left/
 {\vphantom {L \varepsilon }} \right.
 \kern-\nulldelimiterspace} \varepsilon }}^*\!\left( {\varepsilon \bm X,\Omega ,{\bm x_t}} \right)\!,}  {\mathcal R^*_{{L \mathord{\left/
 {\vphantom {L \varepsilon }} \right.
 \kern-\nulldelimiterspace} \varepsilon }\! }\!\left( {\varepsilon \bar {\bm X},{\bm x_t}} \right)} \right)\! \le\! c_{ak}\varepsilon,
  \end{align}
 where $c_{ak}$ is a positive constant.
 Therefore, it can be inferred that $d_{H1}(T_{k+1})\le c_{ak}\varepsilon$.

 As for the averaged system, let us further consider the Hausdorff metric between $ {\mathcal R_{{{L \mathord{\left/
 {\vphantom {L \varepsilon }} \right.
 \kern-\nulldelimiterspace} \varepsilon }} }^*\left( {\varepsilon \bar {\bm X},{\rm{cl}}\mathcal R_{{T_k}}^*\left( {\varepsilon   {\bm X},\Omega ,{\bm x_0}} \right)} \right)}$ and $\mathcal R^*_{{L \mathord{\left/
 {\vphantom {L \varepsilon }} \right.
 \kern-\nulldelimiterspace} \varepsilon }  }\left( \! {\varepsilon \bar {\bm X} ,\mathcal R^*_{ T_k }\left( \! {\varepsilon \bar {\bm X} ,{\bm x_0}}\!  \right)}\!  \right)$, which is correspondingly denoted as $d_{H2}(T_{k+1})$. 
Because $\bm y_{sc}$ is fixed, the Hausdorff metric ${d_H}\!\left(\! {{\rm{gph}}\mathcal S_{t \!+ \!\eta(t) \mathbb B}^ * \left( {\varepsilon \bar  {\bm X},\!{R_1}} \!\right),{\rm{gph}}\mathcal S_{t \!+ \!T\left( t \right) \!+ \!\eta(t)\mathbb B}^ * \left(\! {\varepsilon \bar  {\bm X},\!R_2} \!\right)} \!\right)$ converges asymptotically to $0$ as $t \!\to + \!\infty$, where the initial sets ${R_1}:={\rm{cl}}\mathcal R_{{T_k}}^*\left( {\varepsilon   {\bm X},\Omega ,{\bm x_0}} \right)$ and ${R_2}:=\mathcal R^*_{ T_k }\left( \! {\varepsilon \bar {\bm X} ,{\bm x_0}}\right)$. 
With sufficiently small $\varepsilon_0$, due to the continuity of $\bar {\bm X}$, there exists
\begin{align}
\label{eq:Simplebound}
{d_{H2}}\left( {{T_{k + 1}}} \right) \le \frac{1}{2}{d_H}\left( {{T_k}} \right) + {c_{bk}}\varepsilon ,
\end{align}
for $\forall k \in {\mathbb Z_+ }$, where ${c_{bk}}$ is a positive constant. And there exists ${d_H}\left( {{T_k}} \right) \le {d_{H1}}\left( {{T_k}} \right) + {d_{H2}}\left( {{T_k}} \right), k \in {\mathbb Z_+ }$. Therefore, it can be obtained that
\begin{align}
\label{eq:Totalbound}
{d_H}\left( {{T_{k + 1}}} \right) \le \frac{1}{2}{d_H}\left( {{T_{k }}} \right) + \left( {{c_{ak}} + {c_{bk}}} \right)\varepsilon.
\end{align}
Obviously, from this inequality, we can infer that
\begin{equation}
\label{eq:Totalbound2}
{d_H}\left( {{T_{k + 1}}} \right) \le \left( {\frac{{{2^k} - 1}}{{{2^{k - 1}}}}\left( {{c_{ak}} + {c_{bk}}} \right) + \frac{1}{{{2^k}}}c} \right)\varepsilon.
\end{equation}
Since the coefficient of the $\varepsilon$ in \eqref{eq:Totalbound2} can be set sufficiently small, we can then choose a proper value for $c_{gc}$, \emph{s.t.} 
\begin{equation}
\label{eq:ChooseCgc}
c_{gc}\ge \left( {\frac{{{2^k} - 1}}{{{2^{k - 1}}}}\left( {{c_{ak}} + {c_{bk}}} \right) + \frac{1}{{{2^k}}}c} \right),\nonumber
\end{equation}
for $\forall k \in {\mathbb Z_+ }$. And because all the properties related to \eqref{eq:pointDis}-\eqref{eq:Totalbound2}, do not only hold on any specific instant $T_{k}$, but also hold on $\left[ {t_0,+ \infty} \right)$, which completes the proof of \emph{Theorem~\ref{th:Theorem2}}.
\end{proof}

It is noticeable that there exists a dual relationship between the solutions of the original inclusion and the ones of geometrically averaged system, which is obviously reflected in the conclusions of \emph{Theorem~\ref{th:Theorem1}} and \emph{Theorem~\ref{th:Theorem2}}. Colloquially, the approximation is reciprocal, that is, the original orbits are also approximations of the geometric averaged ones, hence, we give the following corollary.

\begin{Corollary}\label{co:Corollary1}With \emph{Condition A} satisfied, if the original system \eqref{eq:ProbForm} is AGTIS \emph{w.r.t.} 
a critical point $\bm y_{sc} \in D \subset {\mathbb R^n}$, then there exists ${\varepsilon _0}>0$, \emph{s.t.} for $\forall \varepsilon  \in \left( {0,{\varepsilon _0}} \right],$ $\forall \bm x_0 \in D$, and $\forall t \in \mathbb R_{{\ge {t_0}}}$, the inequality \eqref{eq:Th2} still holds.
\end{Corollary}


It is noteworthy that, if all the involved solutions are absolutely continuous as required in \emph{Definition~\ref{def:Rf-solution}},
we can obtain similar results in the sense of pointwise convergence, since the solutions contained in a specific Rf-solution are equicontinuous.
However, this paper focuses on retaining the geometric pattern rather than tracking a specific trajectory.
Moreover, according to the work in \cite{Graph}, the graphical convergence also helps in extending the results for hybrid systems.


\Revised{\begin{Definition}[Strongly \& weakly forward invariance]
The bounded set $\mathcal M \subset D$ is called strongly forward invariant \emph{w.r.t.} the Rf-solutions ${\mathcal S}^{*}_{\mathbb R_{{\ge {t_0}}}}(\bm X,\cdot)$, if for every compact nonempty $R_0 \subset \mathcal M$, the time cross-section $\mathcal R_t^ * \left( {\bm X,{R_0}} \right) \subset \mathcal M$ for $\forall t \in \mathbb R_{\ge t_0}$. And $\mathcal M \subset D$ is called weakly forward invariant \emph{w.r.t.} the Rf-solutions ${\mathcal S}^{*}_{\mathbb R_{{\ge {t_0}}}}(\bm X,\cdot)$, if for every compact nonempty $R_0 \subset \mathcal M$, $\mathcal R_t^ * \left( {\bm X,{R_0}} \right) \cap \mathcal M \ne \emptyset$ for $\forall t \in \mathbb R_{\ge t_0}$.
\end{Definition}}


\Revised{\begin{Theorem}[Set convergence]\label{th:Theorem3}With \emph{Condition A} satisfied, if the geometrically averaged system \eqref{eq:avProbForm} is AGTIS \emph{w.r.t.}
the uniformly bounded Rf-solution $\mathcal S^*_{\mathbb R_{{\ge {t_0}}}}\left( \! {\varepsilon \bar {\bm X} ,R_0^s}\!  \right)$, whose cross-section asymptotically converges to the bounded strongly forward invariant set $\mathcal M_{\bar X}\subset D$, 
%
then there exists ${\varepsilon _0}>0$, \emph{s.t.} for $\forall \varepsilon  \in \left( {0,{\varepsilon _0}} \right]$, $\forall \bm x_0 \in D$, and $\forall t \in \mathbb R_{{\ge {t_0}}},$ 
\begin{align}
\label{eq:Th3}
{d_H}\!\left( {{\rm{cl}}\mathcal R^*_t\left( {\varepsilon \bm X,\Omega ,\bm{x}_0} \right),\mathcal R^*_{  t  }\left( \! {\varepsilon \bar {\bm X},\bm{x}_0}\!  \right)}  \right)\! \le \!  c_{gs},
\end{align}
where $c_{gs}$ is a positive constant.
\end{Theorem}}

\begin{proof}
\Revised{ Choose constants $c_{g1},c_{g2}\in \mathbb R_+$,  $\mathcal M_{g1} : = \mathcal M_{\bar X}+ (c_{g1}+c_{g2}) \mathbb B \subset D$, $\bm x_0 \notin \mathcal M_{g1}$.
Due to ${d_H}\left( {{\mathcal M_{\bar X}}, \bm x_{g1}} \right) \ge c_{g1}+c_{g2}, \forall \bm x_{g1} \in \partial {\mathcal M_{g1}}$, and the absolute continuity of the Rf-solution $\mathcal S^*_{\mathbb R_{{\ge {t_0}}}}\left( \! {\varepsilon \bar {\bm X} ,\bm{x}_0}\!  \right)$, starting from $\partial {\mathcal M_{g1}}$, every Rf-solution with cross-section contained in $\mathcal R^*_t\left( \! {\varepsilon \bar {\bm X} ,\partial {\mathcal M_{g1}}}\!  \right)$ cannot reach ${\mathcal M_{\bar X}}$ in trivial time.
Therefore, due to the asymptotic convergence to the invariance set $\mathcal M_{\bar X}$, it takes less than infinite time, \emph{i.e.} finite time, say, $L_g$, for every selection from the Rf-solution $\mathcal S^*_{\mathbb R_{{\ge {t_0}}}}\left( \! {\varepsilon \bar {\bm X} ,R_0^s}\!  \right)$ to travel from $R_0^s$ to $\mathcal M_{\bar X}+ c_{g1}\mathbb B$. 
And it takes less than $L_a$ length of time for ${d_H}\left( {\mathcal R_t^ * \left( {\varepsilon \bar{\bm X},{\bm x_0}} \right),\mathcal R_t^ * \left( {\varepsilon \bar {\bm X},R_0^s} \right)} \right)$ to become no bigger than $c_{g2}$.
After $L/\varepsilon$ time evolution, larger than $\max \left(L_g, L_a\right)$ by choosing proper $\varepsilon_0$, 
${d_H}\!\left( {\rm{cl}}{\mathcal R^*_t\left( {\varepsilon \bm X,\Omega ,{\bm x_0}} \right),\mathcal R^*_{  t  }\left( \! {\varepsilon \bar {\bm X} ,{\bm x_0}}\!  \right)}  \right)\! \le \!  c\varepsilon_0,$ for $\forall t \in \left[ {t_0,{t_0+L \mathord{\left/
 {\vphantom {t_0+L /\varepsilon }} \right.
 \kern-\nulldelimiterspace} \varepsilon }} \right]$, with the help of \emph{Lemma~\ref{Le:Still}}. And ${\rm{cl}}\mathcal R^*_t\left( {\varepsilon \bm X,\Omega ,{\bm x_0}} \right) \subset \mathcal M_{g1} + c{\varepsilon _0}\mathbb B$, at $ t = t_0+{L \mathord{\left/
 {\vphantom {L /\varepsilon }} \right.
 \kern-\nulldelimiterspace} \varepsilon }$. This inequality does not necessarily hold when $t$ passes ${t_0+L \mathord{\left/
  {\vphantom {t_0+L /\varepsilon }} \right.
  \kern-\nulldelimiterspace} \varepsilon }$, although $\mathcal R^*_{  t  }\left( \! {\varepsilon \bar {\bm X} ,{\bm x_0}}\!  \right)\subset {\rm T}\mathcal M_{g1}$ for $\forall t \ge t_0+L/\varepsilon$.  }

\Revised{However, for every $\bm{x}_\ast$ chosen from the set $\mathcal M_{g1}\! +\! c{\varepsilon _0}\mathbb B$, such that, \emph{w.r.t.} $\mathcal S^*_{\mathbb R_{{\ge \!{t_0\!+\!L\!\mathord{\left/
  {\vphantom {t_0+L / \varepsilon }} \right.
  \kern-\nulldelimiterspace} \varepsilon }}}}\left( {\varepsilon \bm X,\Omega ,\bm{x}_\ast} \right)$, the corresponding Rf-solution $\mathcal S^*_{\mathbb R_{{\ge {t_0+L / \varepsilon}}}}\left( \! {\varepsilon \bar {\bm X} ,\bm{x}_\ast}\!  \right)$ can be selected and further investigated. 
There is a bounded compact set  $\mathcal M_{g2} \!\subset \!D$ which satisfies ${\rm{cl}}\mathcal R^*_{  t  }\left( \! {\varepsilon \bar {\bm X},\mathcal M_{g1} \!+ \!c{\varepsilon _0}\mathbb B}\!  \right) + \!c{\varepsilon _0}\mathbb B\!\subset \!\mathcal M_{g2}$ which contains ${\rm{cl}}\mathcal R^*_{  t  }\left( \! {\varepsilon \bar {\bm X},\bm{x}_\ast}\!  \right)$. 
Due to the asymptotic graphical convergence to $\mathcal S^*_{\mathbb R_{{\ge {t_0}}}}\left( \! {\varepsilon \bar {\bm X} ,R_0^s}\!  \right)$, the asymptotic convergence of $\mathcal S^*_{\mathbb R_{{\ge {t_0}}}}\left( \! {\varepsilon \bar {\bm X} ,R_0^s}\!  \right)$ and the absolute continuity of the Rf-solutions of $\varepsilon \bar {\bm X}$, the cross-section $\mathcal R^*_{  t  }\left( \! {\varepsilon \bar {\bm X},\bm{x}_\ast}\!  \right)$ converges to $\mathcal M_{g1}$ after a time interval of length $L \mathord{\left/
 {\vphantom {t_0+L \varepsilon }} \right.
 \kern-\nulldelimiterspace} \varepsilon$, for a sufficiently small ${\varepsilon _0}$, which further indicates ${\rm{cl}}\mathcal R^*_t\left( {\varepsilon \bm X,\Omega,{\bm x_\ast}} \right) \subset \mathcal M_{g1} + c{\varepsilon _0}\mathbb B$, for $\forall t \in \left[ {t_0,{t_0+L \mathord{\left/
 {\vphantom {t_0+L \varepsilon }} \right.
 \kern-\nulldelimiterspace} \varepsilon }} \right]$.
The preceding process from $\mathcal M_{g1} + c{\varepsilon _0}\mathbb B$ to $\mathcal M_{g2}$ and then back to $\mathcal M_{g1} + c{\varepsilon _0}\mathbb B$ can be indefinitely repeated.
Therefore, ${\rm{cl}}\mathcal R^*_t\left( {\varepsilon \bm X,\Omega,\bm{x}_0} \right) \subset \mathcal M_{g2}$ and $\mathcal R^*_{  t  }\left( \! {\varepsilon \bar {\bm X},\bm{x}_0}\!  \right) \subset \mathcal M_{g1}$ for $\forall t \in \mathbb R_{{\ge {t_0}}}$, with $\mathcal M_{g1}$ and $\mathcal M_{g2}$ both being bounded, and containing the same bounded set $\mathcal M_{\bar X}\subset D$, which verifies the inequality \eqref{eq:Th3}.}
\end{proof}

According to the similar rationale yielding \emph{Corollary~\ref{co:Corollary1}}, and based on \emph{Theorem~\ref{th:Theorem3}}, the following corollary is obtained.
\Revised{\begin{Corollary}\label{co:Corollary2}
With \emph{Condition A} satisfied, and the system \eqref{eq:ProbForm} is AGTIS \emph{w.r.t.} a specific Rf-solution $\mathcal S^*_{\mathbb R_{{\ge {t_0}}}}\left( \! {\varepsilon  {\bm X}, \Omega, \bm{x}_0}\! \right)$. Moreover, this Rf-solution asymptotically converges to the bounded strongly forward invariant set $\mathcal M_{X}\subset D$. Then $\exists {\varepsilon _0}>0$, \emph{s.t.} for $\forall \varepsilon  \in \left( {0,{\varepsilon _0}} \right],$ $\forall \bm x_0 \in D$, and $\forall t \in \mathbb R_{{\ge {t_0}}},$ the inequality \eqref{eq:Th3} still holds.
\end{Corollary}}

\begin{Remark}
Generally, neither the asymptotic nor the exponential stability is invariant under arbitrary time or phase reparameterization. However, these stabilities are preserved under coordinate transformation from time to phase, when the rotation speed is bounded and has a definite direction, keeping the forward unicity or achieving the so-called orientation preserving \cite{Pico-2013}.
\end{Remark}
\Revised{\section{Geometric Periodic Pattern}}

In this section, we search for the existence of periodic solutions, and the specific criteria which guarantee the geometrically convergences by exploiting contraction analyses.
Since we focus on the subtle periodic pattern rather than averaging behavior in this section, we abnegate the small constant $\varepsilon$, and the system \eqref{eq:ProbForm} can then be rewritten as
\begin{align}
\label{eq:ProbFormAuto}
\begin{array}{*{20}{c}}
   {\dot {\bm x} \in  \bm X_\Phi\left( {{ \phi, \bm x}} \right),}  \\
\end{array}{\rm{     }}\begin{array}{*{20}{c}}
   {{\bm x}\left( t_0 \right) = {{\bm x}_0},}  \\
\end{array}
\end{align}
where $\bm X_\Phi: \mathbb R_+ \times {\mathbb R^n} \to \rm{comp}\left( {{\mathbb R^n}} \right)$ is a convex, compact, upper semicontinuous in $\phi$, Lipschitz continuous in $\bm x$, set-valued map, and formulated as:
\begin{equation}
\label{operator}
{\bm X_\Phi }\left( {\phi ,\bm x} \right) = \bigcup\limits_{{1 \mathord{\left/
 {\vphantom {1 {{k_\omega }}}} \right.
 \kern-\nulldelimiterspace} {{k_\omega }}} \in \Omega \left( \bm x \right)} {{k_\omega }\bm X\left( {\phi ,\bm x} \right)},
 \nonumber
\end{equation}
where $k_\omega \in \mathbb R_+$ is bounded, with $D$ being its compact, path-connected domain as well, and the derivative $\dot \star = {{{\rm{d}}\star} \mathord{\left/
 {\vphantom {{{\rm{d}}\star} {{\rm{d}}\phi}}} \right.
 \kern-\nulldelimiterspace} {{\rm{d}}\phi}}$. The set-valued map $\bm X_\Phi$ preserves the compactness, boundedness, existence and uniqueness \emph{w.r.t.} its solution, which are all inherited from $\bm X$, and the evolution induced by $\bm X_\Phi$ is strongly forward invariant. 
Noteworthy is that the system in this setting is considered with the phase pre-synchronized, which indicates that the convergences and the contractions are considered along the phase changing, and the distances between trajectories or solutions are always compared between nearby phases or just with identical phases.

To facilitate further discussion, it is necessary to define the periodic Rf-solutions. For the original system \eqref{eq:ProbForm}, the time cross-section of the periodic Rf-solution ${\mathcal S}^{*}_{t}(\bm X,R_0)$ satisfies that ${\mathcal R}^{*}_{t}(\bm X,R_0) = {\mathcal R}^{*}_{t+\tau}(\bm X,R_0) \subset D$ for some $\tau \in \mathbb R_+$, with $R_0$ being a compact nonempty set in $D$. The periodic Rf-solution \emph{w.r.t.} the phase-based system \eqref{eq:ProbFormAuto} satisfies that ${\mathcal R}^{*}_{\phi}(\bm X,R_0) = {\mathcal R}^{*}_{\phi+2\pi}(\bm X,R_0) \subset D$. 

The following theorem intuitively uses the boundedness to investigate the periodic Rf-solution existence problem. 
\begin{Theorem}
\label{th:PeriodicRf} Consider the system defined in \eqref{eq:ProbFormAuto},
if $\phi_0 \in \mathcal S^1$, the compact nonempty set $R_0 \subset D$ satisfies that ${\mathcal R}^{*}_{\phi}( \bm X_\Phi,R_0) \subset {\mathcal R}^{*}_{\phi+\varphi}( \bm X_\Phi,R_0)$ for some $\varphi\in \mathbb R_+$, and $\bigcup_{\phi \ge {\phi_0}} {\mathcal R}^{*}_{\phi}( \bm X_\Phi,R_0)$ is closed and bounded, then there exists a compact set $R_\infty$, \emph{s.t.}  ${\mathcal R}^{*}_{\phi}( \bm X_\Phi,R_\infty) = {\mathcal R}^{*}_{\phi+\varphi}( \bm X_\Phi,R_\infty)$.
\end{Theorem} 

This theorem can be proved similarly as the proof of \emph{Theorem 5.2}, given in \cite{PanasyukDyn}, by applying Arzel\`{a}-Ascoli theorem.

In \emph{Theorem \ref{th:PeriodicRf}}, the closeness is self-evident due to the definition of the Rf-solution,
while the boundedness can be endowed by constraining the system evolution direction within the Bouligand contingent cones \cite{Blanchini, Aubin} \emph{w.r.t.} some compact path-connected subset of $D$. For a compact level set $S \subset D$, with $\bm x \in D$, the Bouligand contingent cone can be formulated in the neighborhood of $S$ as
\begin{align}
\label{eq:TangentCone}
{\hat T_S}\left( \bm x \right) = \left\{ {\bm v \in \mathbb R^n \left| {\mathop {\lim \inf }\limits_{h \searrow {0}} \underline d{{\left( {\bm x + h\bm v,S} \right)} \mathord{\left/
 {\vphantom {{\left( {x + hv,S} \right)} h}} \right.
 \kern-\nulldelimiterspace} h} = 0} \right.} \right\},
\end{align}
which is nontrivial only when $\bm x$ locates on $\partial S$, since ${\hat T_S}\left( \bm x \right) = \mathbb R^n, \forall \bm x \in {\rm Int} S$ and ${\hat T_S}\left( \bm x \right) = \emptyset$ for $\forall \bm x \notin S$, with ${\rm Int} S$ being the interior of $S$. The cone contains all the vectors shifted to the origin, which, if emanating from $\bm x$, point to where can be contained in $S$ with arbitrarily small $h$.
Furthermore, it is also obviously feasible that the boundedness can also be verified by analyzing averaged systems as the scenario discussed in \emph{Theorem \ref{th:Theorem2}} and \emph{Theorem \ref{th:Theorem3}}.
Besides, considering the Rf-solution with a presupposed period of $2\pi$, we can define the Poincar\'e map $\mathcal P_{2\pi}^{ {\bm X_\Phi }}:\mathcal S^1 \times \rm{comp}\left( {{\mathbb R^n}} \right) \to \mathcal S^1 \times \rm{comp}\left( {{\mathbb R^n}} \right)$, which indicates that $\mathcal P_{2\pi}^{ {\bm X_\Phi }}(\!\phi, R_0\!)\!=\!( \phi + 2\pi, {\mathcal R}^{*}_{\phi+2\pi}(\! \bm X_\Phi,R_0\!)\!)$, $R_0 \!\subset\! D$.
The Poincar\'e map $\mathcal P_{2\pi}^{ {\bm X_\Phi }}$ is called contractive \emph{w.r.t.} single point, if for $\forall (\phi,\bm x_0), (\phi,\bm x_1)   \in Q$, $d_H({\mathcal R}^{*}_{\phi+2\pi}( \bm X_\Phi, \bm x_0),{\mathcal R}^{*}_{\phi+2\pi}( \bm X_\Phi, \bm x_1)) \le k d(\bm x_0,\bm x_1), k \in [0,1)$.
And according to \emph{Theorem~5} in \cite{NadlerMult}, since $(D,d)$ being a complete metric space, if $\mathcal P_{2\pi}^{ {\bm X_\Phi }}$ is contractive \emph{w.r.t.} single point,
then $\mathcal P_{2\pi}^{ {\bm X_\Phi }}$ has a fixed point $\bm x \in D$ as long as $D$ is closed and bounded, \emph{s.t.} $\bm x \in  {\mathcal R}^{*}_{2\pi}( \bm X_\Phi,\bm x))$, equivalently, the phase-based system admits a single-valued periodic solution. Furthermore, even if $\mathcal P_{2\pi}^{ {\bm X_\Phi }}$ is contractive \emph{w.r.t.} single point, based on the fact that the property \eqref{eq:unionA} holds in the phase-based system, and also based on \emph{Lemma~1} and \emph{Theorem~2} in \cite{NadlerMult}, the completeness of the metric space $({\rm{comp}}\left( D \right),d_H)$ can also lead to the existence of a fixed point $R_\infty \in {\rm{comp}}\left( D \right)$, \emph{s.t.} $R_\infty =  {\mathcal R}^{*}_{\phi_0+2\pi}( \bm X_\Phi,R_\infty)$ with the boundary condition $(\phi_0,R_\infty)$, and similarly ${\mathcal R}^{*}_{\phi}( \bm X_\Phi,R_\infty) = {\mathcal R}^{*}_{\phi+2\pi}( \bm X_\Phi,R_\infty)$, which indicates the existence of a $2\pi$-periodic Rf-solution.
In the phase coordinate $\phi$, the state space can be rendered as a forward geodesically complete manifold $\mathcal M_\Phi$, sharing the same topology with $D$, and additionally endowed with a differential structure.
And the system evolution characteristics usually vary in different fibers ${{\rm T}_{\bm x}} {{\mathcal M_\Phi }}$ of the tangent bundle ${{\rm T}} {{\mathcal M_\Phi }}$ over $\forall {\bm x} \in {{\mathcal M_\Phi }}$. In order to diversify the choice of the scalar $\phi$, the metric should not always obey quadratic restriction as the Riemannian metric does \cite{Chern}.
The Finsler structure is more generic to describe the distance between the periodic patterns with different specific forms by integral, and also suitable to describe their convergences through contraction analysis \cite{Lohmiller,Forni}.
The key observation of contraction analysis for periodic system is provided in \cite{SontagCS}. For single-valued continuous systems, asymptotically stable periodic orbits are assured to exist if the infinitesimally contracting vector field is periodic.
The subsequent analyses are to provide the set-valued version of this assertion regarding the periodic Rf-solutions.


\begin{Definition}[Incremental stability \cite{Forni}] Consider all the Carath\'eodory trajectories ${\bm r_\phi }\left( { {\bm X_\Phi},{\bm x_0}} \right)$, hereafter abbreviated as ${\bm r_\phi }\left( {{\bm x_0}} \right)$, contained in the Rf-solutions ${\mathcal S}^{*}_{\mathbb R_{{\ge {\phi_0}}}}(\bm X_\Phi,R_0)$ for all the compact nonempty set $R_0 \subset D$ and $\bm x_0 \in R_0$, 
where the angle $\phi$ is accumulated on $\mathbb R_{\ge \phi_0}$, and $D$ is strongly forward invariant. 
In addition, 
 $\phi_0$ is the start phase angle on $\mathcal S$, and $d:{\mathcal M_\Phi } \times {\mathcal M_\Phi } \to {\mathbb R_{ \ge 0}}$.
The phase-based system \eqref{eq:ProbFormAuto} evolving on the manifold $\mathcal M_\Phi$ is 

\romannumeral1) \emph{incrementally stable} ($\delta$-S) on $D$, if there exists a $\mathcal K$ function $\alpha$ \emph{s.t.} for $\forall {\bm x_1},{\bm x_2} \in D$, $\forall \phi \in {\mathbb R_{ \ge {\phi_0}}}$, $d\left( {{\bm r_\phi }\left( {{\bm x_1}} \right),{\bm r_\phi }\left( {{\bm x_2}} \right)} \right) \le \alpha \left( {d\left( {{\bm x_1},{\bm x_2}} \right)} \right)$;

\romannumeral2) \emph{incrementally asymptotically stable} ($\delta$-AS) on $D$, if it is incrementally stable, and for $\forall {\bm x_1},{\bm x_2} \in D$, $\forall \phi \in {\mathbb R_{ \ge {\phi_0}}}$, then $\mathop {\lim }_{\phi  \to \infty } d\left( {{\bm r_\phi }\left( {{\bm x_1}} \right),{\bm r_\phi }\left( {{\bm x_2}} \right)} \right) = 0$;

\romannumeral3) \emph{incrementally exponentially stable} ($\delta$-ES) on $D$, if $\exists k > 1,\lambda > 0$, for $\forall {\bm x_1},{\bm x_2} \in D$, $\forall \phi \in {\mathbb R_{ \ge {\phi_0}}}$, \emph{s.t.} $d\left( {{\bm r_\phi }\left( {{\bm x_1}} \right), {\bm r_\phi }\left( {{\bm x_2}} \right)} \right) \le k{e^{ - \lambda \left( {\phi  - {\phi _0}} \right)}}d\left( {{\bm x_1},{\bm x_2}} \right)$.
\end{Definition}

\begin{Remark}
It is obvious that, for $\forall  R_0 \in {\rm comp}(D)$, the incremental stability indicates that, if there exists an Rf-solution ${\mathcal S}^{*}_{\mathbb R_{{\ge {\phi_0}}}}(\varepsilon\bm X_\Phi,  R_0)$ that all other Rf-solutions converge to, then it must be unique and single-valued, since the contraction is imposed uniformly on ${{\rm T}} {{\mathcal M_\Phi }}$ such that the limit phase cross-section convex hull is not compressible in any direction.
\label{Re:SuitGraph}
\end{Remark}

\begin{Definition}[Finsler structure \cite{Bao}]
\label{def:Finsler structure}
The Finsler structure of a $C^{\infty}$ smooth manifold ${{\mathcal M }}$ is a non-negative function $F:{\rm T}\mathcal M\to \mathbb R_{\ge 0}$ defined on its tangent bundle, which possesses the following properties:

\romannumeral1) (\emph{Regularity}) $F$ is continuous on the tangent bundle ${\rm T}\mathcal M$, furthermore, $F$ is positive and $C^{\infty}$ smooth on the tangent bundle without the zero section ${\rm T}\mathcal M\backslash 0$; 

\romannumeral2-a) (\emph{Positive homogeneity}) $F\left(\!{\bm x,\lambda \delta \bm x}\!\right)\!=\!\lambda\!F\left( {\bm x,\delta \bm x} \right)$, for $\forall \lambda> 0$ and $\forall \left( {\bm x,\delta \bm x} \right)\in {\rm T}\mathcal M$; 

\romannumeral3) (\emph{Strong convexity}) the Hessian matrix of $F^2$ is positive definite at $\forall \left( {\bm x,\delta \bm x} \right) \in {\rm T}\mathcal M \backslash 0$, \emph{s.t.} $F\left( {\bm x,\delta {\bm x_1} + \delta {\bm x_2}} \right) < F\left( {\bm x,\delta {\bm x_1}} \right) + F\left( {\bm x,\delta {\bm x_2}} \right)$, for $\forall \left( {\bm x,\delta \bm x_1} \right), \left( {\bm x,\delta \bm x_2} \right)\in {\rm T}\mathcal M$, and $\delta \bm x_1 \ne \lambda \delta \bm x_2$, for $\forall \lambda \in \mathbb R$.
\end{Definition}

\begin{Definition}[Finsler distance]
Consider the Finsler manifold $\left(\mathcal M, F\right)$, and the reparameterized curve $\bm \gamma \left( s \right):\left[ {0,1} \right] \to \mathcal M$ being piecewise $C^1$, with $\frac{{{\rm{d}}\bm \gamma \left( s \right)}}{{{\rm{d}}s}} \in {T_{\bm \gamma \left( s \right)}} \mathcal M$. The Finsler distance computing along the curve can be formulated as
\begin{align}
d_F\left( {{\bm x_0},{\bm x_1}} \right): = \mathop {\inf }\limits_{\Gamma \left( {{\bm x_0},{\bm x_1}} \right)} \int_0^1 {F\left( {\bm \gamma \left( s \right),\frac{{{\rm{d}}\bm \gamma \left( s \right)}}{{{\rm{d}}s}}} \right){\rm{d}}s},\nonumber
\end{align}
where $\Gamma$ denotes the collection of all the piecewise $C^1$ curves $\bm \gamma\left( s \right)$ which satisfies $\bm \gamma \left( 0 \right) = {\bm x_0}$, $\bm \gamma \left( 1 \right) = {\bm x_1}$. 
\end{Definition}

Due to the existence of anisotropy, $d_F\left( {{\bm x_0},{\bm x_1}} \right)$ may be different from $d_F\left( {{\bm x_1},{\bm x_0}} \right)$, \emph{i.e.} a quasimetric, and $d_F\left( {{\bm x_0},{\bm x_1}} \right) = d_F\left( {{\bm x_1},{\bm x_0}} \right)$ is true if \romannumeral2-a) in \emph{Definition \ref{def:Finsler structure}} is replaced by:

\romannumeral2-b) (\emph{Absolute homogeneity}) $F\left(\!{\bm x,\lambda \delta \bm x}\!\right)\!=\!|\lambda|F\left( {\bm x,\delta \bm x} \right)$, for $\forall \lambda \in \mathbb R$ and $\forall \left( {\bm x,\delta \bm x} \right)\in {\rm T}\mathcal M$. 
Recall that a function $f:\mathbb R^n \to \mathbb R^m$ is called locally Lipschitz at $\bm x \in \mathbb R^n$, if $\exists \lambda_x, \varepsilon_x \in \mathbb R_+$, \emph{s.t.} $\left\| {f\left( {\bm x'} \right) - f\left( {\bm x''} \right)} \right\| \le {\lambda _x} \max (d_F({\bm x', \bm x''}),d_F({\bm x'', \bm x'}))$ for $\forall \bm x', \bm x'' \in \mathbb B(\bm x,\varepsilon_x)$, 
and $f$ is called locally Lipschitz function if it is locally Lipschitz on $\mathbb R^n$.
And when $f:\mathbb R \times \mathbb R^n \to \mathbb R^m$ has an additional argument of time $t$, it is also called locally Lipschitz at $\bm x \in \mathbb R^n$, if $\exists \varepsilon_x \in \mathbb R_+$, $\exists \lambda_x(t): \mathbb R \to \mathbb R_+$, \emph{s.t.} $\left\| {f\left( {\bm x'} \right) - f\left( {\bm x''} \right)} \right\| \le {\lambda _x}(t)\max (d_F({\bm x', \bm x''}),d_F({\bm x'', \bm x'}))$ for $\forall \bm x', \bm x'' \in \mathbb B(\bm x,\varepsilon_x)$ and $\forall t \in \mathbb R$. 
If $V$ is locally Lipschitz function, which is thus differentiable \emph{a.e.} in the sense of Lebesgue measure due to Rademacher's theorem, and we can then search for the existence of the following inequality to make $V\!\left( {{\bm x},\delta {\bm x}}\! \right)$ a contraction measure:
\begin{align}
\dot V\!\left(\! {\bm x,\delta \bm x} \!\right)\! = \!\frac{{\partial V\!\left( \!{\bm x,\delta \bm x}\! \right)}}{{\partial \bm x}}\dot {\bm x} \!+ \!\frac{{\partial V\!\left(\! {\bm x,\delta \bm x} \!\right)}}{{\partial \delta \bm x}}\delta \dot {{\bm x}} \le \!- \alpha\!\left( {V\!\left( {{\bm x},\delta {\bm x}}\! \right)} \right),
\nonumber
\end{align}
where $\left(\bm x,\delta {{\bm x}}\right) \in {\rm T}\mathcal M_\Phi$,
$\alpha:\mathbb R_{\ge 0} \to \mathbb R_{\ge 0}$ is a locally Lipschitz function, whose specific form are discussed below (see \emph{Lemma~\ref{Le:GGInv}}).
The $C^1$ candidate Finsler-Lyapunov function $V:{\rm T}\mathcal M_\Phi \to \mathbb R_{\ge0}$, is employed to
characterize the contraction behaviors between trajectories, which satisfies
\begin{align}
\label{eq:VIneq}
{c_1}F\left( {{\bm x},\delta {\bm x}} \right)^p \le V\left( {{\bm x},\delta {\bm x}} \right) \le {c_2}F\left( {{\bm x},\delta {\bm x}} \right)^p,
\end{align}
where $c_1, c_2 \in \mathbb R_+$, the positive constant $p \ge 1$, and $F$ is a Finsler structure. 
The anisotropic metric from $R_0$ to $R_1$ on $D$ \emph{w.r.t.} Finsler structure is defined as
\begin{align}
\underline d_{F}\left( {{\bm x_0},{R_1}} \right) =  \inf \limits_{{\bm x_1}\in R_1} d_{F}\left( {{\bm x_0},{\bm x_1}} \right). 
 \label{eq:FinslerHaus}
\end{align}
And the induced Finsler Hausdorff metric is 
\begin{align}
{ d_{F\!H}}(\!{R_0},{R_1}\!)\!= \!\max\! \left\{ {{\!\sup \limits_{{\bm x_0}\in R_0}\underline d_{F\!H}}(\!{\bm x_0},{R_1}\!),\!{\sup \limits_{{\bm x_0}\in R_1}\underline d_{F\!H}}(\!{\bm x_1},{R_0}\!)} \right\}. \nonumber
\end{align}
The incremental stability featured with the graphical paradigm used in \emph{Definition \ref{def:GStability}} and \emph{Definition \ref{def:Asymptotic Gstability} } can then be defined as follows by implementing the Hausdorff metric $ d_{F\!H}$. 
Additionally, $\underline d_{F}$ and $d_{F\!H}$ here can also be seen as its natural extensions on $Q$, similar as the one discussed in \emph{Definition \ref{def:GStability}}.

\begin{Definition}[Graphical incremental stability] 
\label{def:Del-GStability}
Consider the Rf-solutions ${\mathcal S}^{*}_{\mathbb R_{{\ge {\phi_0}}}}(\bm X_\Phi,R_0)$, ${\mathcal S}^{*}_{\mathbb R_{{\ge {\phi_0}}}}(\bm X_\Phi,R_1)$ for any two different compact nonempty sets $R_0, R_1 \in {\rm comp}(D)$, 
with the angle $\phi$ accumulated on $\mathbb R_{\ge \phi_0}$, and $D$ being strongly forward invariant. 
The phase-based system \eqref{eq:ProbFormAuto} is 

\romannumeral1) \emph{graphically incrementally stable} ($\delta$-GS) on $D$, if for any positive constant $\varepsilon_g$,
there exist positive constant $\delta$, 
such that, if we have ${d_{F\!H}}\left(R_0,R_1\right) \le \delta$, then ${d_{F\!H}}( {G\mathcal S_\phi^{\varepsilon_g} \left(R_0\right),G\mathcal S_\phi^{\varepsilon_g} \left(R_1\right)} ) \le \varepsilon_g$, $\forall \phi \in {\mathbb R_{ \ge {\phi_0}}}$, where $G\mathcal S_\phi^{\varepsilon_g} \left(R_0\right):={\rm{gph}}\mathcal S_{\phi + \varepsilon_g \mathbb B}^ * \left( {\bm X_\Phi,{R_0}} \right)$;

\romannumeral2) \emph{graphically incrementally asymptotically stable} ($\delta$-GAS) on $D$, if it is not only $\delta$-GS, but also satisfies $\mathop {\lim }\limits_{\phi \to \infty}  {d_{F\!H}}( \!{G\mathcal S_\phi^{\eta_g(\phi)} \left(R_0\right), G\mathcal S_\phi^{\eta_g(\phi)} \left(R_1\right)} ) = 0$, $\forall \phi \in {\mathbb R_{ \ge {\phi_0}}}$, where the positive Lipschitz continuous function $\eta_g(\phi)$ asymptotically converges to $0$ as $\phi \to \infty$;

\romannumeral3) \emph{graphically incrementally exponentially stable} ($\delta$-GES) on $D$, if $\exists k\! > \!1,\lambda \!> \!0$, \emph{s.t.} ${d_{F\!H}}( \!{G\mathcal S_\phi^{\kappa_g(\phi)} \left(\!R_0\!\right), G\mathcal S_\phi^{\kappa_g(\phi)} \left(\!R_1\!\right)} ) \le k{e^{ - \!\lambda \left( \!{\phi \! - \!{\phi _0}} \!\right)}}{d_{F\!H}}\!\left(R_0,R_1 \right)$, for $\forall \phi \in {\mathbb R_{ \ge {\phi_0}}}$, where the positive Lipschitz continuous function $\kappa_g(\phi)$ exponentially converges to $0$ as $\phi \to \infty$.
\end{Definition}

The following Lemma is to connect the graphical contraction behavior for Rf-solutions and the graphical convergence behavior between Carath\'eodory trajectories.

\begin{Lemma}
\label{Le:NoRfsolu}
Consider any pair of Rf-solutions \emph{w.r.t.} the phase-based system \eqref{eq:ProbFormAuto}, $\mathcal S_{\phi}^ * \left( { \bm X_\Phi,\!{R_0}} \!\right)$ 
and  $\mathcal S_{\phi}^ * \left( { \bm X_\Phi,\!{R_1}} \!\right)$ for $\forall R_0, R_1 \in {\rm{comp}}(D)$ being nonempty, satisfying that

($\ast$) there exists a certain bounded periodic Rf-solution $\mathcal S_{\phi}^ * \left( { \bm X_\Phi,\!{R_\ast}} \!\right)$, where $R_\ast \in {\rm{comp}}(D)$ is convex-valued and path-connected, such that $G\mathcal S_\phi^{\varepsilon_g} \left(R_\ast \right) \subset G\mathcal S_\phi^{\varepsilon_g} \left(R_0\right)$ and $G\mathcal S_\phi^{\varepsilon_g} \left(R_\ast \right) \subset G\mathcal S_\phi^{\varepsilon_g} \left(R_1\right)$, and $R_\ast$ is invariant of every specific choices of $R_0$ and $R_1$.

Define the union $\varepsilon$-graph set as
\begin{align}
&G_\cup\mathcal S_\phi^{\varepsilon_g} \left(R_0, R_1\right) :=G\mathcal S_\phi^{\varepsilon_g} \left(R_0\right)\cup G\mathcal S_\phi^{\varepsilon_g} \left(R_1\right),\nonumber
\end{align}
\noindent where $G\mathcal S_\phi^{\varepsilon_g} \left(R_0\right)$ is abbreviation of ${\rm{gph}}\mathcal S_{\phi \!+ \!\varepsilon_g \mathbb B}^ * \left( {\bm X_\Phi,\!{R_0}} \!\right)$.
Furthermore, consider all the Carath\'eodory trajectories ${\bm r_\phi}$ contained in the set $G_D:={\rm{cl}}(G_\cup\mathcal S_\phi^{\varepsilon_g} \left(R_0, R_1\right) \backslash G\mathcal S_\phi^{\varepsilon_g} \left(R_\ast\right))$.
With \emph{Condition A} satisfied,  $\forall \phi \in {\rm dom} {\bm X _\phi}|_{G_D}$, \emph{i.e.} on the accumulated phase domain in $G_D$, there exist ($\eta_g(\phi)$ or $\kappa_g(\phi)$ is correspondingly used as the superscript instead of $\varepsilon_g$):
\begin{enumerate}[i)]
\item  $\underline d_{F}( {{(\phi,\bm r_\phi)  },G\mathcal S_\phi^{\varepsilon_g} \left(R_\ast\right)} )$ is 0-S $\Leftrightarrow $ \eqref{eq:ProbFormAuto} is $\delta$-GS;

\item $\underline d_{F}( {{(\phi,\bm r_\phi) },G\mathcal S_\phi^{\eta_g(\phi)} \left(R_\ast\right)} )$ is 0-AS $\Leftrightarrow $ \eqref{eq:ProbFormAuto} is $\delta$-GAS;

\item  $\underline d_{F}( {{(\phi,\bm r_\phi) },G\mathcal S_\phi^{\kappa_g(\phi)} \left(R_\ast\right)} )$ is 0-ES $\Leftrightarrow $ \eqref{eq:ProbFormAuto} is $\delta$-GES,
\end{enumerate}


\noindent where 0-S, 0-AS, 0-ES indicate stability, asymptotic stability and exponential stability \emph{w.r.t.} the origin, respectively.
\end{Lemma}

\begin{proof}
In addition to forward complete solutions, there also exist incomplete maximal absolutely continuous solutions restricted to the region $G_D$, where the maximal solution is the solution cannot be extended on $G_D$.
These incomplete solutions can only escape this phase-varying region by penetrating the sub-manifold $\partial{\mathcal R}^{\ast}_{\phi}( \bm X_\Phi,R_\ast)$, rather than $\partial{\mathcal R}^{*}_{\phi}( \bm X_\Phi,R_0)$ nor $\partial{\mathcal R}^{*}_{\phi}( \bm X_\Phi,R_1)$, which is due to the completeness of the Rf-solution \emph{w.r.t.} the state space. And the trajectory after penetration will not return into $G_D$ as a consequence of the strongly forward invariance of $\mathcal S_{\phi}^ * \left( { \bm X_\Phi,\!{R_\ast}} \!\right)$.

The proof of the sufficient assertions ``$\Rightarrow$'': Because $R_\ast \subset R_0 \cap R_1$,  $\underline d_{F}\left( {{\bm x},G_\cap\mathcal S_\phi^{ _g} \left(R_0, R_1\right)} \right) = 0 $ for $\forall \bm x \in G\mathcal S_\phi^{\varepsilon_g} \left(R_\ast\right)$, or $G\mathcal S_\phi^{\eta_g(\phi)} \left(R_\ast\right)$, $G\mathcal S_\phi^{\kappa_g(\phi)} \left(R_\ast\right)$, respectively, where 
\begin{align}
&G_\cap\mathcal S_\phi^{\varepsilon_g} \left(R_0, R_1\right) :=G\mathcal S_\phi^{\varepsilon_g} \left(R_0\right)\cap G\mathcal S_\phi^{\varepsilon_g} \left(R_1\right).\nonumber
\end{align} 
Moreover, the aforementioned incomplete solutions can be extended by a natural method, such that they can take any viable solution which coincides with the penetration point at the penetration instant and always evolves on the Rf-solution boundary $\partial{\mathcal S}^{\ast}_{\mathbb R_{{\ge {\phi_0}}}}( \bm X_\Phi,R_\ast):=\bigcup_{\phi \in \mathbb R_{{\ge {\phi_0}}}}(\phi,\partial{\mathcal R}^{\ast}_{\phi}(\bm X_\Phi,R_\ast))$ as its subsequent part of the solution after the penetration.
The feasibility of this extension relies on the fact that, through any point on $\partial{\mathcal S}^{\ast}_{\mathbb R_{{\ge {\phi_0}}}}( \bm X_\Phi,R_\ast)$, there passes at least one Carath\'eodory trajectory always evolving on this boundary. 
Let us demonstrate this point.
Since there always passes at least one local Carath\'eodory trajectory everywhere on the manifold, by matching adjacent local trajectories, we can find at least one permanent trajectory. 
Suppose that there is a point $\left(\phi,\bm x\right)$ that possesses no such local Carath\'eodory trajectory, then there is a neighbourhood $\left( \phi,\bm x \right) + {\varepsilon _c} \mathbb B$ in $\partial{\mathcal S}^{\ast}_{\mathbb R_{{\ge {\phi_0}}}}( \bm X_\Phi,R_\ast)$ about this point that possesses no such local Carath\'eodory trajectory with a sufficiently small positive constant ${\varepsilon _c}>0$, otherwise this point is isolated and will be neglected owing to the ``almost everywhere'' description in \emph{Definition~\ref{def:AC-solution}}. In other word, every two points $(\phi_0,\bm x_0), (\phi_1,\bm x_1) \in (\left( \phi,\bm x \right) + {\varepsilon _c} \mathbb B) \cap \partial{\mathcal S}^{\ast}_{\mathbb R_{{\ge {\phi_0}}}}( \bm X_\Phi,R_\ast)$ satisfy 
${\bm x_1} \notin {\bm x_0} +  \left| {{\phi _0} - {\phi _1}} \right|{\bm X_\Phi},{\bm x_0} \notin {\bm x_1} +  \left| {{\phi _0} - {\phi _1}} \right|{\bm X_\Phi}$, and due to the Lipschitz continuity of $\bm X_\Phi$, every two points $(\phi_0,\bm x_0), (\phi_1,\bm x_1) \in (\left( \phi,\bm x \right) + {\varepsilon _c} \mathbb B) \cap \partial{\mathcal S}^{\ast}_{\mathbb R_{{\ge {\phi_0}}}}( \bm X_\Phi,R_\ast)  + {\varepsilon _d} \mathbb B$ also satisfy the same property, for a sufficiently small positive constant ${\varepsilon _d}>0$, which indicates no Carath\'eodory trajectory in ${\mathcal S}^{\ast}_{\mathbb R_{{\ge {\phi_0}}}}( \bm X_\Phi,R_\ast)$ can reach $\left(\phi,\bm x\right)$, and thus leads to a contradiction. Therefore, the extension exists for every possible incomplete solutions.
Moreover, it is immediately obvious that, on $\partial{\mathcal R}^{\ast}_{\phi}( \bm X_\Phi,R_\ast)$, $\underline d_{F}( {{(\phi,\bm r_\phi)  },G\mathcal S_\phi^{0} \left(R_\ast\right)} ) \equiv 0$, which indicates 0-S, 0-AS, and 0-ES simultaneously \emph{w.r.t.} any $\varepsilon_g$, $\eta_g(\phi)$ or $\kappa_g(\phi)$. We are now in the position to state that the left side condition given in \romannumeral1), \romannumeral2) and \romannumeral3) can be extended to the region $\forall \phi \in {\mathbb R_{ \ge {\phi_0}}}$, rather than only the domain of $\bm X _\phi$ restricted to $G_D$, $\forall \phi \in {\rm dom} {\bm X _\phi}|_{G_D}$.
Based on \emph{Lemma~\ref{Le:Exists}}, for $\forall \bm x_0 \in D$, noting that all the solutions in $G_D$ are only restrictions of solutions in $\mathcal S_{\phi}^ * \left( { \bm X_\Phi,\!{R_0}} \!\right)$ and  $\mathcal S_{\phi}^ * \left( {\bm X_\Phi,\!{R_1}} \!\right)$, there exists at least one single-valued absolutely continuous solution $\bm r_\phi$, such that $\bm r _{\phi_0} = \bm x_0$. 
Moreover, we have ${d_{F\!H}}( {G\mathcal S_\phi^{\varepsilon_g} \left(R_0\right),G\mathcal S_\phi^{\varepsilon_g} \left(R_1\right)} ) \le {d_{F\!H}}( {G\mathcal S_\phi^{\varepsilon_g} \left(\!R_0\!\right),G\mathcal S_\phi^{\varepsilon_g} \left(\!R_\ast\!\right)} ) + {d_{F\!H}}( {G\mathcal S_\phi^{\varepsilon_g} \left(\!R_1\!\right),G\mathcal S_\phi^{\varepsilon_g} \left(\!R_\ast\!\right)} )$.
Therefore, due to the equivalence discussed in \emph{Lemma~\ref{Le:Exists}} and \emph{Lemma~\ref{Le:Still}}, the convergences of all the single-valued solutions (in the extended sense) imply the convergences of their corresponding $\varepsilon$-graphs in the sense of ${d_{F\!H}}( {G\mathcal S_\phi^{\varepsilon_g} \left(\!R_0\!\right),G\mathcal S_\phi^{\varepsilon_g} \left(\!R_\ast\!\right)})$ or ${d_{F\!H}}( {G\mathcal S_\phi^{\varepsilon_g} \left(\!R_1\!\right),G\mathcal S_\phi^{\varepsilon_g} \left(\!R_\ast\!\right)})$ and further lead to the $\delta$-GS, $\delta$-GAS, and $\delta$-GES between ${\mathcal S}^{*}_{\mathbb R_{{\ge {\phi_0}}}}( \bm X_\Phi,R_0)$ and ${\mathcal S}^{*}_{\mathbb R_{{\ge {\phi_0}}}}( \bm X_\Phi,R_1)$.  
The sufficient assertions ``$\Rightarrow$'' are proved.

The necessary assertions ``$\Leftarrow$'' are more straightforward to prove by considering their converse-negative propositions: 
the condition that ``$\underline d_{F}( {{(\phi,\bm r_\phi)  },G\mathcal S_\phi^{\varepsilon_g} \left(R_\ast\right)} )$ is not 0-S'' obviously leads to that ``\eqref{eq:ProbFormAuto} is not $\delta$-GS'', 
due to the Hausdorff distance definition, which shall apply to \romannumeral2) and \romannumeral3) as well.
\end{proof}

\begin{Remark}
\label{Re:NotContain}
The condition ($\ast$) in \emph{Lemma~\ref{Le:NoRfsolu}} makes the conclusion $\delta$-GS, $\delta$-GAS, and $\delta$-GES therein relatively limited than those in \emph{Definition \ref{def:Del-GStability}}. However, this condition prevents the circumstance in \emph{Remark \ref{Re:SuitGraph}}, and on the contrary, admits set-valued periodic Rf-solution along with incremental stability. Albeit with the limitation, the forthcoming generalization is within reach. For instance, we can use this Lemma to investigate the contraction behaviors between sets $\mathcal S_{\phi}^ * \left( { \bm X_\Phi,\!{R_0}} \!\right)$ 
and  $\mathcal S_{\phi}^ * \left( { \bm X_\Phi,\!{R_1}} \!\right)$, which only satisfy ${R_0} \not\subset {R_\ast}$ and ${R_1} \not\subset {R_\ast}$, by investigating Rf-solutions with their initial conditions as $\overline{\rm{co}}\left( {{R_0} \cup {R_\ast}} \right)$ and $\overline{\rm{co}}\left( {{R_1} \cup {R_\ast}} \right)$, respectively.
\end{Remark}

In order to investigate the differential characteristics of the non-smooth system, we further invoke the Clarke's generalized gradient \cite{Clarke} to analyze $V \left( {\bm x,\delta \bm x} \right)$, which is defined as
\begin{align}
\label{eq:Ggradient}
{{{\partial _C}f\left(\!\bm x \!\right)}} = {\rm{co}}\left\{ {\mathop {\lim }\limits_{i \to \infty } \nabla f\left( {{\bm x_{i}}} \right)\left| {{\bm x_{i}} \to \bm x,{\bm x_{i}} \notin S \cup {\Omega _f}} \right.} \right\},
\end{align}
where ${\rm{co}}$ is the convex hull operation, and $f:\mathbb R^n \to \mathbb R$ is a locally Lipschitz function, ${\Omega _f}$ is the set of points where $f$ fails to be differentiable, $S$ denotes an arbitrary set of zero measure, and $\bm x_{i}$ denotes any sequence converging to $\bm x$, and also avoiding the aforementioned two sets.
The generalized gradient does not exist where $f$ is not locally Lipschitz, but emerges as a nonempty compact convex set where it exists.
\begin{Lemma}
\label{Le:GGInv}
Consider the phase-based system \eqref{eq:ProbFormAuto}, if there exists a regular\footnote{Please refer to \cite{Cortes} for the definition of regular function. Commonly used regular functions include smooth functions and convex functions.} and locally Lipschitz
candidate Finsler-Lyapunov function $V \left( {\bm x,\delta \bm x} \right)$ satisfying the inequality \eqref{eq:VIneq}, and in local coordinate, there exists
\begin{align}
\label{eq:KerIneq}
\sup {\mathcal L_{{\bm X_\Phi}}}V\left( {\bm x,\delta \bm x} \right) \le  - \alpha \left( {V\left( {\bm x,\delta \bm x} \right)} \right),
\end{align}
for $\forall \left(\bm x,\delta {{\bm x}}\right) \in {\rm T}\mathcal M_\Phi$, where the $\sup$ operation is performed locally at $\left( {\bm x,\delta \bm x} \right)$,
and the set-valued Lie derivative is
\begin{align}
\label{eq:Kers}
 {\mathcal L_{{\bm X}_\phi}}V\left( {\bm x,\delta \bm x} \right): = & \frac{{\partial_C V\left( {\bm x,\delta \bm x} \right)}}{{\partial \bm x}}{\bm X_\Phi}\left( {\bm x, \phi} \right)\nonumber\\
 &+   \frac{{\partial_C  V\left( {\bm x,\delta \bm x} \right)}}{{\partial \delta \bm x}}\frac{{\partial {\bm X_\Phi}\left( {\bm x, \phi} \right)}}{{\partial \bm x}}\delta \bm x,
\end{align}
then 
 the phase-based system is 
\begin{enumerate}[i)]
\item $\delta$-S on $\mathcal M_\Phi$, if we can select $\alpha \left( s \right) = 0$ for $s \ge 0$;
\item $\delta$-AS on $\mathcal M_\Phi$, if $\alpha$ can be selected from class~$\mathcal K$ function;
\item $\delta$-ES on $\mathcal M_\Phi$, if $\alpha \left( s \right) := \lambda s$ for $s \ge 0$, with $\lambda$ being a positive constant.
\end{enumerate}


\end{Lemma}
 
Based on \emph{Theorem 1} and \emph{Remark 2} in \cite{Forni}, \emph{Lemma~\ref{Le:GGInv}} can be deduced, which owes its provability to the fact that absolutely continuous solutions cannot realize measurable states change in infinitesimal length of evolution. To be more specific, the distance continuous decrease between any two solutions preserves, since the contraction behavior \eqref{eq:KerIneq} exists \emph{a.e.} along the semi-flow of the system, and \emph{a.e.} along the geodesic curve, or in the relaxed sense, along a certain piecewise smooth curve. 
\begin{Remark}
\label{Re:Weak1}
In \emph{Lemma~\ref{Le:GGInv}}, we expose the incremental stability actually within the strong sense, where all the solutions are required to be $\delta$-S, $\delta$-AS, or $\delta$-ES. And one can also investigate within the weak sense, where only at least one solution needs to possess the properties \romannumeral1-\romannumeral3) therein, by replacing \eqref{eq:KerIneq} with
\begin{align}
\label{eq:KerIneqinf}
\inf {\mathcal L_{{\bm X_\Phi}}}V\left( {\bm x,\delta \bm x} \right) \le  - \alpha \left( {V\left( {\bm x,\delta \bm x} \right)} \right). 
\end{align}
Moreover, the candidate Finsler-Lyapunov function may obey a time-varying form (or the phase-varying form), rendered as $V \left( {t, \bm x,\delta \bm x} \right)$ (or alternatively $V \left( {\phi, \bm x,\delta \bm x} \right)$), where the set-valued Lie derivative becomes
\begin{align}
 {\mathcal L\!_{{\bm X}_\phi}}V\left( {t,\bm x,\delta \bm x} \right): = & \frac{{\partial_C \!V\left( {t,\bm x,\delta \bm x}\! \right)}}{{\partial t}} + \frac{{\partial_C \!V\left({t,\bm x,\delta \bm x} \!\right)}}{{\partial \bm x}}{\bm X_\Phi}\left( {\bm x,\phi} \right)\nonumber\\
 &+   \frac{{\partial_C \! V\left( {t,\bm x,\delta \bm x} \right)}}{{\partial \delta \bm x}}\frac{{\partial {\bm X_\Phi}\left( {\bm x,\phi} \right)}}{{\partial \bm x}}\delta \bm x. \nonumber
\end{align}
\end{Remark}

\begin{Remark}
\label{Re:CalcofGG} Let us demonstrate the calculation of the set-valued Lie derivative. Although the definition \eqref{eq:Kers} enjoys the beauty of consistency \emph{w.r.t.} its single-valued version, we still need the following version to complete the set calculation:
\begin{align}
 {\mathcal L_{{\bm X_\Phi}}}V\left( {\bm x,\delta \bm x} \right) = \left\{ {a \in {\mathbb R}\left| {\exists \bm v \in {\hat{\bm X}_\Phi}{\rm{,}}} \right.} \right. \nonumber\\ 
 \left. {{\rm{s}}{\rm{.t}}{\rm{. }~}a = \left\langle {\bm  v,\bm  \zeta } \right\rangle ,{\rm{ for~}}\forall \bm  \zeta  \in {\partial _C}V\left( {\bm x,\delta \bm x} \right)} \right\}, \nonumber
\end{align}
where ${\hat{\bm X}_\Phi}\!: = \!\left[ {{\bm X_\Phi}^T\!\left( {\phi, \bm x} \right),{{\left[ {{{\partial {\bm X_\Phi}\left( {\phi, \bm x} \right)\cdot\delta \bm x} \mathord{\left/
 {\vphantom {{\partial {\bm X_\Phi}\left( {\phi, \bm x} \right)\cdot\delta \bm x} {\partial \bm x}}} \right.
 \kern-\nulldelimiterspace} {\partial \bm x}}} \right]}^T}} \right]$.
When $V\left( {\bm x,\delta \bm x} \right)$ is continuously differentiable at $\left( {\bm x,\delta \bm x} \right)$, then ${\mathcal L_{{\bm X_\Phi}}}V\left( {\bm x,\delta \bm x} \right) = \{ \nabla V\left( \bm x \right) \cdot {\bm v}, \bm v \in {\hat{\bm X}_\Phi}\}$.
By virtue of \emph{Lemma~1} in \cite{Bacciotti1999}, the set-valued Lie derivative ${\mathcal L_{{\bm X_\Phi}}}V\left( {\bm x,\delta \bm x} \right)$ exists almost everywhere, and is a closed and bounded interval in $\mathbb R$ and possibly empty, as long as $V \left( {\bm x,\delta \bm x} \right)$ is regular and locally Lipschitz, and $\bm X_\Phi$ satisfies condition \romannumeral1) of \emph{Theorem~\ref{th:Theorem1}}. 
According to \cite{KamalapurkarB}, since $V \left( {\bm x,\delta \bm x} \right)$ is selected to be regular, $\sup {\mathcal L_{{\bm X_\Phi}}}V\left( {\bm x,\delta \bm x} \right)$ can also be relaxed as $\mathop {\min }_{\bm \zeta  \in {\partial _C}V\left( {\bm x,\delta \bm x} \right)} \mathop {\max }_{\bm v \in {{\rm D}_{{\bm X_\Phi }}}} \left\langle {\bm v,\bm \zeta } \right\rangle$.
And we adopt the convention $\max\left( \emptyset \right) = - \infty$.
Moreover, the partial derivative ${{{\partial {\bm X_\Phi}\left( {\phi, \bm x} \right)} \mathord{\left/
 {\vphantom {{\partial {\bm X_\Phi}\left( {\phi, \bm x} \right)} {\partial \bm x}}} \right.
 \kern-\nulldelimiterspace} {\partial \bm x}}}$ can be calculated via the contingent derivative ${\rm D}{\bm X_\Phi }\left( {\phi, \bm x,\bm y} \right)$ at $(\phi ,\bm x,\bm y) \in {\rm gph} \bm X_\Phi$, defined geometrically from the choices of contingent cones \eqref{eq:TangentCone} to the graphs, such that, in local coordinate,
\begin{align}
{\rm{gph}} {{\rm D}{\bm X_\Phi }\left( {\phi, \bm x, \bm y} \right)} = \hat T_{{\rm{gph}} {{\bm X_\Phi }} }\left( {\phi, \bm x, \bm y} \right), \nonumber
\end{align}
where $\bm x \in \mathcal M_\Phi$ and $\bm y \in {\rm T}_{\bm x}\mathcal M_\Phi$. In the light of \cite{Ledyaev2007}, the notion of contingent cone can be adapt to Finsler manifold:
\begin{align}
&\hat T_{{\rm{gph}} {{\bm X_\Phi }} }\left( {\phi, \bm x, \bm y} \right)= \{ {\bm v_\phi =1, \bm v_x \in {\rm T}_{\bm x}\mathcal M_\Phi,  \bm v_y \in {\rm T}_{\bm y}{\rm T}_{\bm x}\mathcal M_\Phi | } \nonumber\\ & {{\mathop {\lim \inf }\limits_{h \to {0^ + }} \underline d{{\left(\left(\phi+h, {\rm exp}_{\bm x}\left(h \bm v_x \right), {\rm exp}_{\bm y}\left(h \bm v_y \right) \right),{\rm{gph}} {{\bm X_\Phi }} \right)} \mathord{\left/
 {\vphantom {{\left( {x + hv,{\rm{gph}} {{\bm X_\Phi }}} \right)} h}} \right.
 \kern-\nulldelimiterspace} h} = 0} } \}. \nonumber
\end{align} 
where $h \in \mathbb R_+$, ${\rm exp}_{\star}(h \bm v)$ is the exponential map which is the end of the geodesic emanating from $\star$ and along $h \bm v$ with a length equaling $|h \bm v|$.
Furthermore, if $\bm v \in {\rm D}{\bm X_\Phi }\left( {\phi, \bm x, \bm y} \right)( 0, \delta {\bm x})$, and the ambient space of $\mathcal M_\Phi$ is available, we have 
\begin{align}
\mathop {\lim \inf }\limits_{h \searrow 0 , \bm u' \to  \delta {\bm x}} \underline d\left( {\bm v,\frac{{\bm X_\Phi\left(\phi, {\rm exp}_{\bm x}(h \bm u') \right) - \bm y}}{h}} \right) = 0,\nonumber
\end{align}
where $\bm u' \in {\rm T}_{\bm x}\mathcal M_\Phi$, $h \in \mathbb R_+$. Moreover, according to the Hopf-Rinow theorem \cite{Bao}, the universal existence of the exponential map can be ensured by the forward geodesic completeness of the manifold $\mathcal M_\Phi$.
\end{Remark}

Different from the uniform contraction analysis discussed in \emph{Lemma~\ref{Le:GGInv}}, the promising method to attack the set-valued version contraction problem is to encode the condition in \emph{Lemma~\ref{Le:NoRfsolu}} into the Finsler structure. Based on these two points, the following Lemma is established.

\begin{Lemma}
\label{Le:GGssInv}
Consider the phase-based system \eqref{eq:ProbFormAuto}, if there exists a regular and locally Lipschitz
candidate Finsler-Lyapunov function $V \left( {\bm x,\delta \bm x} \right)$ satisfying the inequality \eqref{eq:VIneq}, and in local coordinate, there exists
\begin{align}
\sup {\mathcal L_{{\bm X_\Phi}}}V\left( {\bm x,\delta \bm x} \right) \le  - \alpha \left( {V\left( {\bm x,\delta \bm x} \right)} \right),\label{eq:CarIneqset}
\end{align}
for all $\left(\bm x,\delta {{\bm x}}\right) \in \mathcal{I}_{c}(\phi):= \mathbb{R}_{\ge \phi_0} \to {\rm T}\mathcal M_\Phi$, where $\pi (\mathcal{I}_{c}(\phi)) =\mathcal M_\Phi \backslash {\bm {\mathcal {F}}}_{c} (\phi)$,
and the customized funnel (donut, $\mathcal{I}_{c}(\phi):= \mathcal S^1 \to {\rm T}\mathcal M_\Phi$, if it is periodic) ${\bm {\mathcal {F}}}_{c}: \mathcal{S}^1 \to {\rm{comp}}(\mathcal M _\phi)$ being a compact, convex-valued, bounded, and path-connected, set-valued map uniformly contained in $D$, which also satisfies
\begin{align}\label{eq:ContainDF}
{\bm X_\Phi }\left( {\phi ,\bm x} \right) \subseteq {\rm D}{{\bm {\mathcal {F}}}_{c}}\left( {\phi ,\bm x} \right)\left( 1 \right),
\end{align}
for \emph{a.e.} $\phi \in \mathbb R_{\ge \phi_0}$ and $\forall \bm x \in D$, where ${\rm D}{{\bm {\mathcal {F}}}_{c}}$ is the contingent derivative of the funnel ${{\bm {\mathcal {F}}}_{c}}$, $\phi_0$ is the initial phase angle of the solution, and $\bm x_{in}$ is the initial state,
where $\pi: {\rm T}\mathcal M_\Phi \to \mathcal M_\Phi$ is the canonical projection,
then the phase-based system is 
\begin{enumerate}[i)]
\item stable \emph{w.r.t.} ${\bm {\mathcal {F}}}_{c}(\phi)$ in $\mathcal M_\Phi$, if $\alpha \left( s \right) = 0$ for $s \ge 0$;
\item asymptotically stable \emph{w.r.t.} ${\bm {\mathcal {F}}}_{c}(\phi)$ in $\mathcal M_\Phi$, if $\alpha$ can be selected from class~$\mathcal K$ function;
\item exponentially stable \emph{w.r.t.} ${\bm {\mathcal {F}}}_{c}(\phi)$ in $\mathcal M_\Phi$, if $\alpha \left( s \right) := \lambda s$ for $s \ge 0$, with the positive constant $\lambda \in \mathbb R_+$.
\end{enumerate}


\end{Lemma}

\begin{proof}
Consider an arbitrary Carath\'eodory trajectory ${\bm r_\phi(\phi_0,\bm x_e) }$ outside the funnel ${\bm {\mathcal {F}}}_{c}(\phi)$ for every $\phi$ accumulated on $\mathbb R_{\ge \phi_0}$.
According to \cite{Donchev2005,Aubin1987}, the subset relationship \eqref{eq:ContainDF} indicates that the system $({{\bm {\mathcal {F}}}_{c}(\phi)}, \bm X_\Phi)$ is strongly forward invariant, which suggests that every solution ${\bm r_\phi(\phi_0,\bm x_{in}) }$ starts from $(\phi_0, \bm x_{in}) \in {\rm gph} {{\bm {\mathcal {F}}}_{c}(\phi)}$ remains in ${{\bm {\mathcal {F}}}_{c}(\phi)}$.
Furthermore, the length of an oriented curve connecting ${\bm r_\phi(\phi_0,\bm x_e) }$ to ${\bm r_\phi(\phi_0,\bm x_{in}) }$ can be computed by
\begin{align}
d_F\left( {{{\bm r_\phi(\phi_0,\bm x_e) }},{\bm r_\phi(\phi_0,\bm x_{in}) }} \right) = \int_0^1 {F\left( {\bm \gamma_\phi \left( s \right),\frac{{{\partial}\bm \gamma_\phi \left( s \right)}}{{{\partial}s}}} \right){\rm d}s},\nonumber
\end{align}
where $\bm \gamma_\phi \left( 0 \right) ={\bm r_\phi(\phi_0,\bm \gamma_{\phi_0} \left( 0 \right))}$, $\bm \gamma_\phi \left( 1 \right) = {\bm r_\phi(\phi_0,\bm \gamma_{\phi_0} \left( 1 \right)) }$, $\bm \gamma_{\phi_0} \left( 0 \right) = \bm x_e$ and $\bm \gamma_{\phi_0} \left( 1 \right) = \bm x_{in}$. 
Thus, there exists a scalar function $\mu\left(\phi\right) \in \left(0, 1\right)$, such that $\bm \gamma_{\phi} \left( \mu \left(\phi\right)\right) \in \partial {\bm {\mathcal {F}}}_{c}\left(\phi\right)$. And it is obvious that
\begin{align}
\underline d_{F}\left( {{\bm r_\phi(\phi_0,\bm x_e)},{{\bm {\mathcal {F}}}_{c}(\phi)}} \right) \le  \int_0^{\mu\left(\phi\right)} {F\left( {\bm \gamma_{\phi} \left( s \right),\frac{{{\partial}\bm \gamma_{\phi} \left( s \right)}}{{{\partial}s}}} \right){\rm d}s}.\nonumber
\end{align}
Since the system $({{\bm {\mathcal {F}}}_{c}(\phi)}, \bm X_\Phi)$ is strongly forward invariant, $\mu\left(\phi\right)$ is monotonically non-decreasing, further there exists
\begin{align}
&\underline d_{F}\left( {{\bm r_\phi(\phi_0,\bm x_e)},{{\bm {\mathcal {F}}}_{c}(\phi)}} \right) \nonumber\\  \le &{c_1}^{ - \frac{1}{p}} \int_0^{\mu\left(\phi_0\right)} {V\left( {\bm \gamma_{\phi} \left( s \right),\frac{{{\partial}\bm \gamma_{\phi} \left( s \right)}}{{{\partial}s}}} \right)^{\frac{1}{p}}{\rm d}s},\nonumber
\end{align}

The remaining analysis is then tantamount to that of \emph{Lemma~\ref{Le:GGInv}} and \emph{Theorem 1} in \cite{Forni}, which further shows that $\underline d_{F}({{\bm r_\phi(\phi_0,\bm x_e)},{{\bm {\mathcal {F}}}_{c}(\phi)}})$ is 0-S, 0-AS, or 0-ES, respectively.
\end{proof}

\begin{Remark}
\label{Le:Practical}
From a practical perspective of maintaining periodic patterns, \emph{Lemma~\ref{Le:GGssInv}} indicates that we can develop a periodic funnel satisfying \eqref{eq:ContainDF} but not necessarily along the evolution of $\bm X_{\Phi}$ as the target, which enhances desired trajectory design freedom comparing to elaborately selecting an  exact single-valued solution subject to the physical constraints.
\end{Remark}

\begin{Remark}
Similar to \emph{Remark~\ref{Re:Weak1}}, in the weak sense, one can use \eqref{eq:KerIneqinf} instead of \eqref{eq:CarIneqset}, and use the condition
${\bm X_\Phi }\left( {\phi ,\bm x} \right) \cap D{{\bm {\mathcal {F}}}_{c}}\left( {\phi ,\bm x} \right)\left( 1 \right) \ne \emptyset$ 
instead of \eqref{eq:ContainDF}.
If the following three conditions are satisfied: ${ {\overline {\bm{\mathcal { F}}}}_{c}} := \mathop {\lim \sup }_{\phi  \to  + \infty } {{\bm {\mathcal {\bm F}}}_{c}}\left( \phi  \right)$, the phase angle $\phi$ completely relies on the states $\bm x$, and $d_{F\!H}\left({{ \bm{\mathcal { F}}}_{c}}(\phi),{ {\overline {\bm{\mathcal { F}}}}_{c}}\right)$ is 0-S, 0-AS, or 0-ES in the phase coordinate, then in the strong sense, we can use ${\bm X_\Phi }\left( {\phi ,\bm x} \right) \subseteq {\hat T_{ { \overline{\mathcal {\bm F}}}_{c}}}(\bm x)$ instead of \eqref{eq:ContainDF}, or in the weak sense, we can use ${\bm X_\Phi }\left( \phi, \bm x \right) \cap {\hat T_{ { \overline{\mathcal {F}}}_{c}}}(\bm x) \ne \emptyset$, for \emph{a.e.} $\phi \in \mathbb R_{\ge \phi_0}$ and $\forall \bm x \in D$.
\end{Remark}

\begin{Theorem}[Single-valued contraction]
\label{th:Finnal1} Consider the system defined in \eqref{eq:ProbFormAuto}, if there exists a phase angle $\phi_0 \in \mathcal S^1$ and a compact nonempty set $R_0 \subset D$ satisfying that ${\mathcal R}^{*}_{\phi}( \bm X_\Phi,R_0) \subset {\mathcal R}^{*}_{\phi+2\pi}( \bm X_\Phi,R_0)$, and $\overline{\mathcal R}_0 := \bigcup_{\phi \ge {\phi_0}} {\mathcal R}^{*}_{\phi}( \bm X_\Phi,R_0)$ is closed and bounded, and there exists a regular and locally Lipschitz
candidate Finsler-Lyapunov function $V \left( {\bm x,\delta \bm x} \right)$ satisfying \eqref{eq:VIneq} and \eqref{eq:KerIneq}, 
the phase-based system is 

\begin{enumerate}[i)]
\item stable \emph{w.r.t.} some specific periodic single-valued solution contained in $\overline{\mathcal R}_0$, if we can set $\alpha \left( s \right) = 0$ for $s \ge 0$;

\item asymptotically stable \emph{w.r.t.} some specific periodic single-valued solution contained in $\overline{\mathcal R}_0$, if $\alpha$ is of class~$\mathcal K$;

\item  exponentially stable \emph{w.r.t.} some specific periodic single-valued solution contained in $\overline{\mathcal R}_0$, if $\alpha \left( s \right) := \lambda s$ for $s \ge 0$, with $\lambda$ being a positive constant.
\end{enumerate}



\end{Theorem}

By implementing \emph{Theorem \ref{th:PeriodicRf}}, we know there exists at least one periodic single-valued solution, and according to \emph{Lemma~\ref{Le:GGInv}} and \emph{Remark~\ref{Re:SuitGraph}}, all the Rf-solutions converges to a certain single-valued solution, which together verify \emph{Theorem \ref{th:Finnal1}}.

\begin{Theorem}[Set-valued contraction]
\label{th:Finnal2} In the system \eqref{eq:ProbFormAuto}, if there exists a phase angle $\phi_0 \in \mathcal S^1$ and a compact nonempty set $R_0 \subset D$ satisfying that ${\mathcal R}^{*}_{\phi}( \bm X_\Phi,R_0) \subset {\mathcal R}^{*}_{\phi+2\pi}( \bm X_\Phi,R_0)$, and $\overline{\mathcal R}_0$ is closed and bounded, and there exists a regular and locally Lipschitz
candidate Finsler-Lyapunov function $V \left( {\bm x,\delta \bm x} \right)$ satisfying the inequalities \eqref{eq:VIneq} and \eqref{eq:CarIneqset}, 
and the funnel ${{\bm {\mathcal {F}}}_{c}}$ is itself a periodic Rf-solution, the phase-based system with initial sets containing ${\bm {\mathcal {F}}}_{c}$ is 
\begin{enumerate}[i)]
\item graphically stable \emph{w.r.t.} the funnel ${{\bm {\mathcal {F}}}_{c}}$, if we can set $\alpha \left( s \right) = 0$ for $s \ge 0$;

\item graphically asymptotically stable \emph{w.r.t.} the funnel ${{\bm {\mathcal {F}}}_{c}}$, if $\alpha$ can be selected from class~$\mathcal K$ function;

\item  graphically exponentially stable \emph{w.r.t.} the funnel ${{\bm {\mathcal {F}}}_{c}}$, if $\alpha \left( s \right) := \lambda s$ for $s \ge 0$, with $\lambda$ being a positive constant.
\end{enumerate}
\end{Theorem} 

Based on \emph{Theorem \ref{th:PeriodicRf}}, we know there exists at least one periodic Rf-solution, and according to \emph{Lemma~\ref{Le:NoRfsolu}} and \emph{Lemma~\ref{Le:GGssInv}}, all the Rf-solutions graphically converge to the periodic Rf-solution ${{\bm {\mathcal {F}}}_{c}}$, which together verify \emph{Theorem \ref{th:Finnal2}}.

\section{Application in Biomimetic Robot System}
Biomimetic robot systems can be described by the hybrid dynamical systems, and the \emph{hybrid zero dynamics (HZD)} method lends itself to locomotion nonlinear feedback controller designs. In this section, we employ the geometrical pattern to attack this problem with the help of the differential inclusion and the incremental stability analysis.
Although the jumps in hybrid systems are not concerned in our previous analyses, the hybrid dynamics analysis is practically available because of the graphical closeness discussed in \emph{Definition~\ref{def:Del-GStability}}. In order to maintain the periodic pattern, similar to \cite{Li2016}, it is essential to prioritize the continuous phase angle over the jump counter, such that we use ${d_{F\!H}}( {G\mathcal S_{\phi,j_0}^{\varepsilon} \left(R_0\right),G\mathcal S_{\phi,j_1}^{\varepsilon} \left(R_1\right)})$ instead of ${d_{F\!H}}( {G\mathcal S_{\phi}^{\varepsilon} \left(R_0\right),G\mathcal S_{\phi}^{\varepsilon} \left(R_1\right)})$ to extend \emph{Definition \ref{def:Del-GStability}}. The Rf-solution interval $\mathcal S_{\phi,j}^{\varepsilon}$ is parameterized by the phase angle $\phi \in \mathbb R_{\ge \phi_0}$ and the jump counter $j \in \mathbb N$. Since there probably exists instantaneous increases in $\phi$ after a jump, we emphasize that only the physically present continuous evolution is taken into account of the interval $\varepsilon$ consumption.

The nonlinear input-affine system used in dynamic walking studies \cite{MorrisB,Chevallereau2009,Hamed,Reher2020}, associated with the discrete instantaneous impact, can be formulated as a hybrid dynamics system: 
\begin{equation}
 \label{eq:Hybridwalk}
\Sigma (t,j):\left\{ \begin{array}{l}
 \dot{\bm x} = \bm f\left( \bm x \right) +  G\left( {\bm x} \right) \bm u,~~~~{\bm x} \notin \mathcal S \\ 
 {{\bm x}^ + } =\bm \Delta \left( {\bm x} \right),~~~~~~~~~~~~~~{\bm x} \in \mathcal S \\ 
 \end{array} \right.
\end{equation}
where the state $\bm x := {\left( {{\bm q^\top},{{\dot {\bm q}}^\top}} \right)^\top}
 \in  \mathcal X \subset {\mathbb R^{2n}},$ the generalized coordinate vector $\bm q\in \mathcal Q$ and $\mathcal X = {\rm T}\mathcal Q \subset {\mathbb R^{n}}$, with ${\rm T} \mathcal Q$ being the tangent bundle of $\mathcal Q$, ${\bm x^ + }(t) := \mathop {\lim }_{\tau  \searrow t} \bm x\left( \tau  \right)$ indicates the next value of the state after the jump, 
 the control input $\bm u \in \mathcal U \subset \mathbb R^m$, the Lipschitz continuous function $\bm f:\mathcal X \to {\rm T} \mathcal X$, and ${\rm T} \mathcal X$ is the tangent bundle of $\mathcal X$, $G:\mathcal X \to \mathbb R ^{2n \times m}$, and the reset map $\bm \Delta:  \mathcal X \to  \mathcal X$. The codimension-1 submanifold termed \emph{guard} or \emph{switching surface}, defined as $\mathcal S:=\left\{\bm x \in \mathcal X | s\left( \bm x \right) = 0 \right\}$, with $s: \mathcal X \to \mathbb R$ satisfying $\nabla_{\bm x} s \left(\bm x\right) \ne 0$, is the space where jump happens, which is naturally chosen as the Poincar\'e section to construct the Poincar\'e map. Several basic assumptions about dynamic walking are resorted to for rigorous analysis, such as being complete but neither eventually continuous nor eventually discrete (\emph{c.f.} \cite{Goebel2011} for their definitions), and possessing no \emph{Zeno} or \emph{beating} behaviors (\emph{c.f.} \cite{Haddad2006} for their definitions).
And the desired output function is given by ${\bm y_d}\left( {{\phi _d}\left( \bm q \right),\bm \alpha } \right)$, with the strictly monotonic phase parameterization function ${\phi _d}: \mathcal Q \to \mathcal S^1$, $\bm \alpha \in \mathbb R^k$ being the associated parameter vector, and $\bm y_d: \mathcal S^1 \times \mathbb R^k \to \mathcal Q$ is the desired vector function to be controlled depending on the phase. 
The \emph{virtual constraint} can then be imposed by zeroing the output $\bm y\left(\bm q,\bm \alpha \right): = {\bm y_a}\left(\bm q \right) - {\bm y_d}\left( {{\phi _d}\left(\bm q \right),\bm \alpha } \right) \equiv \bm 0$, with $\bm y\left(\bm q,\bm \alpha \right)$ being the corresponding difference between actual and desired outputs, which also plays a central role in specifying the \emph{zero dynamics manifold}: $\mathcal Z: = \left\{ \bm x \in \mathcal X | \bm y\left( {\bm q,\bm \alpha} \right) = \bm 0,{\mathcal L_f}\bm y\left( {\bm q, \dot{\bm q},\bm \alpha} \right) = \bm 0 \right\}$. In order to make $\mathcal Z$ forward invariant and attractive, the feedback linearizing controller can be employed:
\begin{align}
 \label{eq:Hybridinput}
 &\bm u\left(\bm x \right) = {\left( {{\mathcal L_G}{\mathcal L_f}\bm y\left( {\bm x,\bm \alpha } \right)} \right)^{ - 1}}\left( { - \mathcal L_f^2\bm y\left( {\bm x,\bm \alpha } \right) - \bm \mu } \right), \\ 
 &\bm \mu  = \frac{1}{\kappa ^2}{K_p}\bm y\left( {\bm x,\bm\alpha } \right) + \frac{1}{\kappa }{K_d}{\mathcal L_f}\bm y\left( {\bm x,\bm \alpha } \right), \nonumber 
 \end{align}
 as long as ${{\mathcal L_G}{\mathcal L_f}\bm y\left( {\bm x,\bm \alpha } \right)} $ is invertible, where $\mathcal L_f$ and $\mathcal L_G$ represent the Lie derivative along $f$ and $G$, respectively. And the feedback gains $K_p, K_d \in \mathbb R^{n \times n}$ are positive definite, making the generated linear system locally exponentially stable. The parameter $\kappa \in \mathbb R_+$ can modulate the rate of convergence. When the impact invariance condition $\bm \Delta(\mathcal S \cap \mathcal Z) \subset \mathcal Z$ is satisfied, the system can then achieve periodic stable walking. 

In the feedback linearizing controller \eqref{eq:Hybridinput}, the stability is restricted to the local and periodic fashion, which narrows the applicability of the designed controller. And the uncertainty of the model and the intermittent unforeseeable interactions with the environment should be described more intrinsically. Thus, consider the differential inclusion hybrid system \eqref{eq:Hybridwalk}:
\begin{equation}
 \label{eq:HybridwalkDiff}
\Sigma^d (t,j):\left\{ \begin{array}{l}
 \dot{\bm x} \in \bm f^d \left( \bm x \right) + G^d \left( {\bm x} \right) \bm u,~~~{\bm x} \notin \mathcal S \\ 
 {{\bm x}^ + } \in \bm \Delta^d  \left( {\bm x} \right),~~~~~~~~~~~~~~{\bm x} \in \mathcal S \\ 
 \end{array} \right.
\end{equation}
where $\bm f^d: \mathcal X \to {\rm comp}\left({\rm T} \mathcal X\right)$, $ G^d: \mathcal X \to {\rm comp}\left(\mathbb R ^{2n \times m}\right)$, and $\bm \Delta^d:  \mathcal X \to  {\rm comp}\left(\mathcal X\right)$ are all nonempty, compact, convex-valued, and upper semi-continuous. 
And the corresponding nominal orbit is also uncertain, which is thus given by the funnel ${\bm {\mathcal {F}}}_{c}(\phi)$, resembling a set-valued version of zero dynamics manifold. By implementing \emph{Theorem \ref{th:Finnal2}}, we demonstrate the corresponding controller design in the rest of this section, and keep the discussion to the minimum for illustrative purpose only.
First, the differential dynamics is 
\begin{equation}
 \label{eq:HybridwalkDiffDelta}
\Sigma _\delta ^d(t,j):\left\{ {\begin{array}{*{20}{c}}
   {\dot \delta {\bm x} \in F_\delta ^d\left( {{\bm x},{\bm u}} \right)\delta {\bm x} + G ^d\left( {\bm x} \right)\delta {\bm u},~{\bm x} \notin \mathcal S}  \\
   {\delta {{\bm x}^ + } \in \bm \Delta _\delta ^d\left( {{\bm x}} \right)\delta {\bm x}, ~~~~~~~~~~~~~~~~~{\bm x} \in \mathcal S}  \\
\end{array}} \right.
\end{equation}
where $F_\delta ^d\left( {{\bm x},{\bm u}} \right):= {{\partial {\bm f^d}} \mathord{\left/
 {\vphantom {{\partial {\bm f^d}} {\partial \bm x}}} \right.
 \kern-\nulldelimiterspace} {\partial \bm x}} + \sum\nolimits_{i = 1}^m {\left( {{{\partial \bm g_i^d} \mathord{\left/
 {\vphantom {{\partial \bm g_i^d} {\partial \bm x}}} \right.
 \kern-\nulldelimiterspace} {\partial \bm x}}} \right)}  \cdot {u_i}$, $g_i^d$ is the $i$\,th column of $G^d$, and $u_i$ is the $i$\,th entry of $\bm u$, moreover, $\bm \Delta _\delta ^d\left( {{\bm x}} \right):= {{\partial {\bm \Delta^d}} \mathord{\left/
 {\vphantom {{\partial {\bm \Delta^d}} {\partial \bm x}}} \right.
 \kern-\nulldelimiterspace} {\partial \bm x}}$. 
Then we use the phase parameterization function ${\phi _d}\left( \bm q \right)$ as the phase angle $\phi$ to describe the dynamics, and invoke the continuous part of \eqref{eq:HybridwalkDiff} to generate ${{{\rm d}\phi } \mathord{\left/{\vphantom {{{\rm d}\phi } {d{\rm d}t}}} \right.\kern-\nulldelimiterspace} {{\rm d}t}}$:
\begin{equation}
{\rm D}_\phi \left( {{\bm x},{\bm u}} \right): = \frac{{{\rm d}\phi }}{{{\rm d}t}} \in {\left[ {\begin{array}{*{20}{c}}
   {{{\left( {{\nabla _{\bm q}}{\phi _d}} \right)}^\top}} & {{\bm 0^\top}}  \\
\end{array}} \right]}\left[ {\bm f^d \left( \bm x \right) + G^d \left( {\bm x} \right) \bm u} \right].
\nonumber
\end{equation}
In phase coordinate, the dynamics \eqref{eq:HybridwalkDiffDelta} can be recast as
\begin{equation}
\Sigma _\delta ^d(\phi,j):\left\{ {\begin{array}{*{20}{c}}
   {\dot \delta {\bm x} \in \bar F_\delta ^d\left( {{\bm x},{\bm u}} \right)\delta {\bm x} + \bar G ^d\left( {\bm x},{\bm u}\right)\delta {\bm u},~{\bm x} \notin \mathcal S}  \\
   {\delta {{\bm x}^ + } \in \Delta _\delta ^d\left( {{\bm x}} \right)\delta {\bm x}, ~~~~~~~~~~~~~~~~~~~~~{\bm x} \in \mathcal S}  \\
\end{array}} \right.\nonumber
\end{equation}
where $\dot \delta {\bm x}$ should be interpreted as ${{{\rm d}\delta {\bm x}} \mathord{\left/
 {\vphantom {{{\rm d}\delta x} {{\rm d}\phi }}} \right.
 \kern-\nulldelimiterspace} {{\rm d}\phi }}$, $\bar F_\delta ^d\left( {{\bm x},{\bm u}} \right):= {\rm D}_\phi^{-1}\left( {{\bm x},{\bm u}} \right) F_\delta ^d\left( {{\bm x},{\bm u}} \right)$, and $\bar G_\delta ^d\left( {{\bm x},{\bm u}} \right):= {\rm D}_\phi^{-1}\left( {{\bm x},{\bm u}} \right) G_\delta ^d\left( {{\bm x}} \right)$. In order to satisfy the inequality \eqref{eq:KerIneq}, taking $\bar f_\delta ^d\left( {{\bm x},{\bm u}} \right):= {\rm D}_\phi^{-1}\left( {{\bm x},{\bm u}} \right) f_\delta ^d\left( {{\bm x}} \right)$, we define the differential controller $\delta \bm u$ as
 \begin{equation}
 \label{eq:deltauBelong}
 \delta \bm u  \in  - {\left( {\bm B{\bm B^\top}} \right)^{ - 1}}\left( {A + \kappa_\alpha \alpha \left( V \left( \bm x, \delta \bm x \right)\right)} \right){\bm B^\top},\\ 
\end{equation}
where $A: = \frac{{\partial V}}{{\partial \bm x}}\left( {{{\bar {\bm f}}^d}\left( \bm x \right) + {{\bar G}^d}\left( \bm x \right)\bm u} \right) + \frac{{\partial V}}{{\partial \delta \bm x}}\bar F_\delta ^d\left( \bm x \right)\delta \bm x$, the row vector $\bm B: = \frac{{\partial V}}{{\partial \delta \bm x}}{{\bar G}^d}\left( \bm x \right)$, and $\kappa_\alpha \in \mathbb R_+$ controls the convergence rate. We can choose a specific value  for $\delta \bm u$ subject to \eqref{eq:deltauBelong}, making $\sup {\partial _C}V\left( {\bm x,\delta \bm x} \right)$ as small as possible, and the limit is denoted as $\delta \bm u_b$. Moreover, \emph{w.r.t.} the funnel ${\bm {\mathcal {F}}}_{c}(\phi)$, it is assumed that we can find a specific control law $\bm u^*(\bm x, \phi)$ to make the continuous evolution invariant in ${\bm {\mathcal {F}}}_{c}(\phi)$ with arbitrary initial condition therein. Then, we use the path integral to calculate the control input $\bm u_b$ according to $\delta \bm u_b$. For $\forall \phi \in {\rm dom}{\bm {\mathcal {F}}}_{c}$, choosing a point $\bm x^* (\phi)$ on the boundary $\partial{\bm {\mathcal {F}}}_{c}(\phi)$, we have 
 \begin{align}
&\bm u_b\left( {\bm c,{\bm u^ * },\phi ,s} \right) \nonumber \\ 
= &{\bm  u^ * }\left( {{\bm  x^ * }\left( \phi  \right),\phi } \right) + \int_0^s {\delta {\bm u_b}\left( {\bm c\left( {\phi ,\mu } \right),\frac{{\partial \bm  c\left( {\phi ,\mu } \right)}}{{\partial \mu }},\bm u_b} \right)} {\rm d}\mu, \nonumber
\end{align}
where $\bm c (\phi,\cdot)$ is a smooth phase-varying curve. This specific curve can be numerically approximated and designed to converge to a geodesic connecting $\bm x$ and $\bm x^*$ \cite{Chaffey2018,Wang2020}.
Then we can use $\bm u_b\left( {\bm c,{\bm u^ * },\phi ,1} \right)$ as the continuous part input.
Moreover, since there exists a state jump in every single step, it may be necessary to determine a start phase angle $\phi_s$ and an end phase angle $\phi_e$ for every step, and we need to elaborately design the funnel so as to satisfy that $\bm \Delta^d  \left({\bm {\mathcal {F}}}_{c}(\phi_e+\varepsilon_{\phi} \mathbb B)+\varepsilon_{\mathcal {F}} \mathbb B\right)\subset {\bm {\mathcal {F}}}_{c}(\phi)$, for $\forall \phi \in \phi_s +\varepsilon_{\phi} \mathbb B$, where $\varepsilon_{\phi} \in \mathbb R_+$ indicates the unachieved part in the phase angle, and $\varepsilon_{\mathcal {F}} \in \mathbb R_+$ indicates the unachieved part in the states, both of which can be chosen according to the convergence behavior achieved by the continuous part controller, that is, as long as the continuous dynamics convergence is sufficiently fast, we can choose arbitrarily small constants.

\section{Conclusion}
In this paper, we study both the average behavior and the retention of the geometrically periodic pattern observed in the geometrically periodic system.
Notwithstanding the illustrative example in robotic dynamic locomotions,
the established theory also has other obvious potential applications in multi-agent systems, morphological dynamics and boundary known uncertain systems by the virtue of its set-valued nature.
There also exist theoretical extensions: \romannumeral1) the averaging method of multiple geometrical labels and the corresponding manifold development, \romannumeral2) the contraction analysis of periodic orbits with different topologies or fuzzy Finsler structures, and \romannumeral3) the extension of this work to the loop space to explore more topological properties of the periodic patterns.

}

\end{document}